\documentclass[review, 3p]{elsarticle}
\pdfoutput = 1
\usepackage[numbers]{natbib}

\usepackage{amssymb}
\usepackage{graphicx}

\usepackage{pgf,tikz}

\usetikzlibrary{snakes}
\tikzstyle{printersafe}=[snake=snake,segment amplitude=0 pt]

\usetikzlibrary{arrows}
\usetikzlibrary{shapes}
\usepackage{url}
\usepackage{amsmath}
\usepackage{caption}
\usepackage{todonotes}
\usepackage[linesnumbered,ruled]{algorithm2e}
\usepackage{pgf, tikz}
\usetikzlibrary{graphs, quotes, graphs.standard}
\usetikzlibrary{matrix}
\usepackage{tabularx}
\usepackage{arydshln,leftidx,mathtools}
\usepackage{setspace}

\usepackage{comment}
\usepackage{mathtools}
\newcommand{\floor}[1]{\left\lfloor #1 \right\rfloor}
\newcommand{\ceil}[1]{\left\lceil #1 \right\rceil}

{\endarray\hspace*{-0.5\arraycolsep}\endBmatrix}
{\endarray\endbmatrix}

\usepackage{enumerate}
\usepackage{bm}
\usepackage[T1]{fontenc}
\usepackage[utf8]{inputenc}
\usepackage[english]{babel}
\usepackage{float}

\usepackage{accents}

\newtheorem{theorem}{Theorem}
\newtheorem{proposition}{Proposition}

\newtheorem{definition}{Definition}

\newtheorem{lemma}{Lemma}

\newenvironment{proof}{ {\bf Proof:}} %{$\Box$}

\journal{European Journal of Operational Research}

\begin{document}
%\onehalfspacing

\begin{frontmatter}

%\title{Integer Linear Programming models for\\Total Coloring and Total Matching problems \tnoteref{label0}}
\title{Total Coloring and Total Matching: Polyhedra and Facets}
%\tnotetext[label0]{This is only an example}

\author{Luca Ferrarini}%\corref{cor1}\fnref{label3}}
\ead{l.ferrarini3@campus.unimib.it}

\author{Stefano Gualandi}
\ead{stefano.gualandi@unipv.it}

\address{Università di Pavia, Dipartimento di Matematica ``F. Casorati''}

%\cortext[cor1]{author}
%\fntext[label3]{I also want to inform about\ldots}
%\fntext[label4]{Small city}

%\address[label4]{Università di Pavia, Dipartimento di Matematica ``F. Casorati''}

\begin{abstract}
A total coloring of a graph $G=(V,E)$ is an assignment of colors to vertices and edges such that neither two adjacent vertices nor two incident edges get the same color, and, for each edge, the end-points and the edge itself receive different colors.
Any valid total coloring induces a partition of the elements of $G$ into total matchings, which are defined as subsets of vertices and edges that can take the same color.
In this paper, we propose Integer Linear Programming models for both the Total Coloring and the Total Matching problems, and we study the strength of the corresponding Linear Programming relaxations.
The total coloring is formulated as the problem of finding the minimum number of total matchings that cover all the graph elements.
%, and we prove that this relaxation is tighter than a natural assignment model.
%
This covering formulation can be solved by a Column Generation algorithm, where the pricing subproblem corresponds to the Weighted Total Matching Problem. 
Hence, we study the Total Matching Polytope.
We introduce three families of nontrivial valid inequalities: vertex-clique inequalities based on standard clique inequalities of the Stable Set Polytope, congruent-$2k3$ cycle inequalities based on the parity of the vertex set induced by the cycle, and even-clique inequalities induced by complete subgraphs of even order. 
We prove that congruent-$2k3$ cycle inequalities are facet-defining only when $k=4$, while the vertex-clique and even-cliques 
are always facet-defining.
Finally, we present preliminary computational results of a Column Generation algorithm for the Total Coloring Problem and a Cutting Plane algorithm for the Total Matching Problem.
%
%For the first problem, we compare the two linear relaxation of the models.
%
%For the latter, we compare the strength of the violated cuts generated by the  method.
%Since the separation problem of the clique inequalities of even order is NP-hard, we get a polyhedral proof of the NP-hardness of the Weighted Total Matching Problem. 
%
%Computational results on random and standard benchmarks show that these problems are indeed challenging for Integer Linear Programming approaches.

%solved with a polyhedral approach and, hence, we introduce a number of valid inequalities that are facet-defining for the Total Matching Polytope.
%

%We propose an Integer Linear Programming approach based on a column generation algorithm where the pricing subproblem consists of solving a Maximum Weighted Total Matching problem.
%
%The pricing subproblem is solved with a polyhedral approach and, hence, we introduce a number of valid inequalities that are facet-defining for the Total Matching Polytope.
%
%Our main result is a new facet-defining cycle inequality which is based on the parity of the cycle.
\end{abstract}


\end{frontmatter}

% \linenumbers


\section{Introduction}\label{sec:intro}
Let us consider a simple and undirected graph $G=(V,E)$ and let $D = V \cup E$ be the set of its elements.
We say that a pair of elements $a,b \in D$ are {\it adjacent} if $a$ and $b$ are adjacent vertices, or if they are incident edges, or if $a$ is an edge incident to a vertex $b$.
If two elements $a,b \in D$ are not adjacent, they are {\it independent}.
A {\it matching} is a subset of edges $M \subseteq E$ such that $e \cap f = \emptyset$ for all $e,f \in M$ with $e \neq f$. 
A matching is called {\it perfect} if it covers all vertices, that is, has size $\frac{1}{2}|V|$. We define $\nu(G) := \max \{ |M| : M \mbox{ is a matching} \}$.
Given a set of colors $K=\{1,\dots,k\}$, a $k$--total coloring of $G$ is an assignment $\phi : D \rightarrow K$ such that $\phi(a) \neq \phi(b)$ for every pair of adjacent elements $a,b \in D$.
Each subset of elements assigned to the same color by $\phi$ defines a \textit{total matching}, that is, a subset $T \subseteq D$ where the elements are pairwise independent.
Hence, a $k$-total coloring induces a partition of the elements in $D$ into $k$ disjoint total matchings.
The minimum value of $k$ such that $G$ admits a $k$-total coloring is called the total chromatic number, and it is denoted by $\chi_{T}(G)$.
A total matching of maximum cardinality is denoted by $\nu_T(G)$.

%%%%%%%%%%%%%%%%%%%%%%%%%%%%%%%%%
The Total Coloring Problem (TCP) consists of finding $\chi_T(G)$. 
It is an NP-hard problem \cite{sanchez1989}, which is studied mainly in graph theory \cite{Yap2006} for the conjecture attributed independently to Vizing \cite{Vizing1964} and Behzad \cite{Behzad1965} that relates $\chi_{T}(G)$ to the maximum degree $\Delta(G)$ of the nodes in $G$. 
The conjecture states that $\chi_T(G) \leq \Delta(G)+2$.
Note that $\Delta(G)+1$ is a trivial lower bound on $\chi_T(G)$. 
Hence, if the conjecture were true, we would be left with the NP-Complete problem of deciding whether $\chi_T(G) = \Delta(G)+1$.
While the conjecture is valid for specific classes of graphs (e.g., see \cite{Vignesh2018}), the conjecture is still open for general graphs\footnote{When submitting this paper, we have found a paper on arxiv claiming a proof for Vizing's conjecture \cite{Murthy2020}.}.
In particular, Vizing's conjecture was proved for cubic graphs, and hence, the total chromatic number of a cubic graph is either 4 or 5.
%
%For this reason, among all the classes of graphs, the cubic graphs have been studied in literature for their major relevance. 
%
In \cite{Brinkmann2015}, the authors raised the question of whether a cubic graph of Type 2 (i.e., with $\chi_T(G) = \Delta(G)+2$) and with girth greater than 4 exists, where the girth is defined as the length of the smallest cycle in the graph.
%
%Their intention is to investigate a possible connection between girth and total chromatic parameters in cubic graphs.
%
Up to now, the question is still open.

The TCP generalizes both the Vertex Coloring Problem, where we have to color only the vertices of $G$, and the Edge Coloring Problem, where instead we have to color only the edges. 
The Vertex Coloring Problem belongs to the list of 21 NP-hard problems introduced by E. Karp in \cite{karp1975}, and it was tackled in the literature by many exact polyhedral approaches (for a recent survey, see \cite{malaguti2010}). 
The most efficient exact approaches to the Vertex Coloring Problem are based on set covering formulations \cite{Mehrotra1996,gualandi2012,malaguti2011,held2012}, where every single set of the covering represents a subset of vertices taking the same color, corresponding hence to a (maximal) stable set of $G$.
Similarly, the best polyhedral approach to the Edge Coloring Problem is based on a set covering formulation, where the edges are covered by (maximal) matchings of $G$ \cite{Nemhauser1991,Lee1993}. 
An interesting alternative formulation for the Edge Coloring Problem is based on a different ILP model based on a binary encoding of the problem variables \cite{Lee2005}. 
While in the literature there are other interesting approaches to graph coloring problems (e.g., branch-and-cut \cite{Mendez2006}, semidefinite programming \cite{Karger1998}, decision diagrams \cite{VanHoeve2020}, constraint satisfiability \cite{Hebrard2020}, memetic algorithms \cite{lu2010}), and other interesting types of coloring problems (e.g., equitable coloring \cite{Lih1998,Mazz}, graph multicoloring \cite{Gualandi2012b}, sum coloring \cite{DelleDonne2020}, selective graph coloring \cite{demange2015}), in this work, we focus on a polyhedral approach to the TCP.

A problem related to TCP is the Total Matching Problem (TMP), where we look for a subset of the elements of $G$ which yield an independent set.
The TMP generalizes both the Matching Problem, where we look for an independent set of edges \cite{Edmonds1965}, and the Stable Set Problem, where instead we look for an independent set of vertices \cite{Rossi2001,Rebennack2011}.
The first work on the TMP appeared in \cite{TotalMatching}, where the authors derive lower and upper bounds on the size of a maximum total matching.
In \cite{Manlove}, D.F. Manlove provides a survey of the algorithmic complexities of the decision problems related to graph parameters.
The author reports that $\nu_{T}(G)$ is NP-complete for bipartite, chordal and arbitrary graphs.
Despite the strong connection with the Matching Problem, the TMP is less studied in the operations research literature.
In particular, significant results are obtained only for structured graphs, as  cycles, paths, full binary trees, hypercubes, and complete graphs, see \cite{Leidner2012}. 
This work presents a first polyhedral study of the TMP deriving several facet-defining inequalities for its polytope.

%%%%%%%%%%%%%%%%%%%%%%%%
The TCP has several practical applications, for instance, in Match Scheduling \cite{SportScheduling}, Network Task Efficiency, and Math Art \cite{Leidner2012}.
As an example of match scheduling, consider the martial art tournament problem, which can be modeled using a {\it tournament graph} $G=(V,E)$ and a set of colors $K$ defined as follows.
We add a vertex $i$ to $V$ for each player, and an edge $\{i,j\}$ to $E$ for each match.
Then, for each time period of the tournament, we introduce a color in $K$.
The assignment of a color $k\in K$ to an edge $\{i,j\}$ represents the scheduled time period of the match between players $i$ and $j$.
The assignment of a color $k \in K$ to a vertex $i$ represents a rest time for player $i$ during the time period associated with color $k$.
Given this graph formulation, no pair of incident edges get the same color because no player can be in two matches at once;
no vertex can be incident to an edge with the same color as the vertex because no player should have a match during his rest time; 
no pair of adjacent vertices should get the same color because no two matched players can leave the stage simultaneously.
Hence, a proper total coloring of the tournament graph represents a feasible scheduling of the tournament, and the total chromatic number represents the minimum number of time periods to schedule the tournament.

%%%%%%%%%%%%%%%%%%%%%%%%%%%%%%%
\paragraph{Our contributions} The main results of this paper are:
\begin{itemize}
    \item A set covering formulation of the TCP based on maximal total matchings, which can be solved by column generation, and which yields a lower bound at least as strong as the lower bound obtained with standard ILP formulation. 
    %Motivated by the pricing subproblem of the covering formulation, we introduce the Maximum Weighted Total Matching Problem.

    \item The definition of three families of nontrivial valid inequalities that we call 
    the {\it vertex-clique} inequalities based on standard clique inequalities of the Stable Set Polytope,
    the {\it congruent-$2k3$ cycle inequalities}, which are based on the parity of the cycle, and the {\it even-clique inequalities}, which are based on complete graphs of even cardinality. 
    We prove that congruent-$2k3$ cycle inequalities are facet-defining only when $k=4$ (see Proposition \ref{cycle4}), while the vertex-clique and even-clique (but not the odd-clique) inequalities are always facet-defining (see Theorems \ref{vertex-clique} and \ref{kappaEven}).
    
    \item Preliminary computational results on the strength of the set covering relaxation with respect to the assignment relaxation, and on the computational strength of the valid inequalities introduced for the Total Matching Polytope.
%
    
%    \item A polyhedral proof of the NP-hardness of the Weighted Total Matching Problem.
    
\end{itemize}

\paragraph{Outline} The outline of this paper is as follows.
In the next paragraph, we fix the notation.
In Section 2, we present two Integer Linear Programming (ILP) formulations of the TCP, and we introduce the Maximum Weighted Total Matching Problem (MWTMP). 
In Section 3, we study the Total Matching Polytope, proposing several basic facet-defining inequalities.
In Section 4, we present three families of nontrivial valid inequalities: vertex-clique inequalities,
%based on standard clique inequalities, 
congruent-$2k3$ cycle inequalities, %based on the parity of the vertex set induced by the cycle, 
and even-clique inequalities. % induced by complete subgraphs of even order. 
Section 5 provides computational results for the two problems, comparing the lower bounds of the two models for the TCP and the strength of the new families of inequalities generated in a cutting-plane framework.
%Using the complexity of the separation problem for even-clique inequalities, we provide a polyhedral proof of the NP-hardness of the Weighted Total Matching Problem.
%
In Section 6, we conclude the paper with a discussion on future works.

\paragraph{Notation}
The graphs considered in this paper are simple and undirected. 
Given a graph $G=(V,E)$, we define $n = |V|$ and $m = |E|$. 
For a vertex $v \in V$, we denote by $\delta(v)$ the set of edges incident to $v$ and by $N_{G}(v)$ the set of vertices adjacent to $v$. 
The degree of a vertex is $|\delta(v)|$, in particular, we denote by $\Delta(G)= \max\{|\delta(v)| \mid v \in V\}$. 
For a subset of vertices $U \subseteq V$, let $G[U]$ be the subgraph induced by $U$ on $G$.
We define $\delta(U):=\{e\in E \mid e=\{u,v\}, u \in U, v \in V \setminus U \}$.

%

%%%%%%%%%%%%%%%%%%%%%%%%%%%%%%%%%%%%%%%%%%%%%%%%%%%%%%%%%%%%%%%%%%%%%%%%%%%%%%%%%%%%%%%%%%%%%%%%%%%%%%%%%%%%%%%%
\section{Total Coloring and Total Matching Models: ILP models}\label{sec:coloring}
In this section, we first present an assignment Integer Linear Programming (ILP) model for the TCP. 
Second, we introduce a stronger set covering formulation based on the idea of covering the elements of the graph by the minimum number of maximal total matchings.
Third, we introduce the Weighted Total Matching Problem (WTMP).
\subsection{Total Coloring: Assignment model} Let $G=(V,E)$ be a graph and let $K$ be the set of available colors, with $|K| \geq \Delta(G)+1$. 
We introduce binary variables $x_{vk} \in \{0,1\}$ for every vertex $v$ and binary variables $y_{ek} \in \{0,1\}$ for every edge $e$ to denote whether they get assigned color $k$.
Besides, we introduce the binary variables $z_k$ to indicate whether any element uses color $k$. 
Using these variables, our assignment ILP model for the TCP is as follows.
\begin{align}
\label{m1:obj}   \chi_{T}(G) \, := z_{IP}^{(A)} \, =  \min \quad & \sum_{k \in K} z_k \\
\label{m1:c1}    \mbox{s.t.} \quad 
    & \sum_{k \in K} x_{vk} = 1 & \forall v \in V \\
\label{m1:c3}    & \sum_{k \in K} y_{ek} = 1 & \forall e \in E \\
\label{m1:c4}    & x_{vk} + \sum_{e \in \delta(v)} y_{ek} \leq z_k & \forall v \in V, \forall k \in K \\
\label{m1:c5}    & x_{vk} + x_{wk} + y_{ek} \leq z_k & \forall e = \{v,w\} \in E, \forall k \in K \\
\label{m1:x}    & x_{vk} \in \{0,1\} & \forall v \in V, \forall k \in K \\
\label{m1:y}    & y_{ek} \in \{0,1\} & \forall e \in E, \forall k \in K.
\end{align}
\noindent The objective function \eqref{m1:obj} minimizes the number of used colors.
Constraints \eqref{m1:c1}--\eqref{m1:c3} ensure that every vertex and every edge get assigned a color. 
Constraint \eqref{m1:c4} enforces that all edges $e$ incident to a vertex $v$, and the vertex $v$ itself, take a different color; at the same time, the constraints guarantee that the corresponding variable $z_k$ is set to 1 whenever color $k$ is used by at least an element of $G$. 
Constraint \eqref{m1:c5} imposes that for each edge $e=\{i,j\}$ at most one element among $\{e,i,j\}$ can take color $k$, and it sets the corresponding variable $z_k$ accordingly. 
If we relax the integrality constraints \eqref{m1:x} and \eqref{m1:y}, we get a Linear Programming relaxation. 
We denote the optimal value of the LP relaxation by $z_{LP}^{(A)}$.

The LP relaxation of model \eqref{m1:obj}--\eqref{m1:y} yields the following lower bound.

\begin{proposition}
Let $G=(V,E)$ be a graph. Then, we have $\chi_{T}(G) \geq z_{LP}^{(A)} \geq \Delta+1$.
\end{proposition}
\begin{proof}
Let $x_{vk}=y_{ek}=\frac{1}{\Delta+1}$ for $k=1, \dots, \Delta+1$,
$z_k = 1$ for $k=1,2,\dots, \Delta+1$ and $x_{vk} = 0, y_{ek} = 0$ for $k > \Delta+1, \forall v \in V, \forall e \in E$. 
Notice that this assignment gives a feasible solution for the LP relaxation of \eqref{m1:obj}--\eqref{m1:y}.
Since $|K| \geq \Delta+1$, the assertion follows immediately. \qed

%{\bf NOTA: discutere che $\Delta$ è sempre maggiore o uguale a 2?}
%
\end{proof}
\subsection{Total Coloring: Set Covering model} 
The assignment model \eqref{m1:obj}--\eqref{m1:y} is easy to write, but it suffers from symmetry issues: any permutation of the color classes indexed by $k$ generates the same optimal solution \cite{Margot2007,Jabrayilov2018}.
To overcome this issue and to get a stronger LP lower bound, we introduce a set covering formulation based on maximal total matchings. 
A total matching is (inclusion-wise) maximal if it is not a subset of any other total matching. 
Note that the number of maximal total matchings in a graph is strictly less than the number of total matchings.

Let $\mathcal{T}$ be the set of all maximal total matchings of $G$.
Let $\lambda_t$ be a binary decision variable indicating if the matching $t \subset \mathcal{T}$ is selected (or not) for representing a color class.
The 0--1 parameter $a_{vt}$ indicates if vertex $v$ is contained in the total matching $t$. 
Similarly, the 0--1 parameter $b_{et}=1$ indicates if edge $e$ is contained in the total matching $t$. 
The following set covering model is a valid formulation for the TCP.
\begin{align}
\label{m2:obj} \chi_{T}(G) \, = z_{IP}^{(C)}\, := \min \quad & \sum_{t \in \mathcal{T}} \lambda_{t} \\
\label{m2:c1} \mbox{s.t.} \quad & \sum_{t \in \mathcal{T}} a_{vt} \lambda_{t} \geq 1 & \forall v \in V\\
\label{m2:c2} &  \sum_{t \in \mathcal{T}} b_{et}\lambda_{t} \geq 1 & \forall e \in E\\
\label{m2:l} & \lambda_{t} \in \{0,1\} & \forall t \in \mathcal{T}.
\end{align}
\noindent Given an optimal solution of the previous problem, whenever an element of $G$ appears in $t>1$ maximal total matchings, it is always possible to recover a proper total coloring by removing the element from $t-1$ of those total matchings.
Note that the covering model has an exponential number of variables, one for each maximal total matching in $G$. 
We denote by $z^{(C)}_{LP}$ the optimum value of the LP relaxation of problem \eqref{m2:obj}--\eqref{m2:l}.

If we introduce the dual variables $\alpha_v$ for constraints \eqref{m2:c1} and the variables $\beta_e$ for constraints \eqref{m2:c2}, we can write the dual of the set covering LP relaxation as follows. 
\begin{align}
\label{m3:primal} z_{LP}^{(C)} := \min \; \quad & \sum_{t \in \mathcal{T}} \lambda_{t} \quad \mbox{ s.t. \eqref{m2:c1}--\eqref{m2:c2}}, \lambda_t \geq 0, \forall t \in \mathcal{T} & (Primal)\\
\label{m3:d1} = \max \quad & \sum_{v \in V} \alpha_{v} + \sum_{e \in E} \beta_{e} & (Dual)\\
\label{m3:d2} \mbox{s.t.} \quad & \sum_{v \in V} a_{vt} \alpha_{v} + \sum_{e \in E} b_{et} \beta_{e} \leq 1 & \forall t \in \mathcal{T}\\
\label{m3:d3} & \alpha_{v}, \beta_{e} \geq 0 & \forall v \in V, \forall e \in E.
\end{align}
For this LP covering relaxation, the following proposition holds.
\begin{proposition}
Let $G=(V,E)$ be a graph. Then, we have $\chi_T(G) \geq z_{LP}^{(C)} \geq \Delta(G)+1$.
\end{proposition}
\begin{proof}
Consider a vertex $v$ of maximum degree, and let $\Delta(G)=k$, where $N_{G}(v):=\{v_{1},...,v_{k}\}$ and $\delta(v):=\{e_1,...,e_{k}\}$. 
Consider the total matching $T_0 := \{v\}$, and the additional $k$ distinct total matchings $T_{i}:=\{v_{i}, e_{i+1}\}$ for all $i=1,\dots,k-1$ and $T_{k}:=\{v_k,e_{1}\}$. Hence, we have $k+1$ total matchings, which can be used to define a feasible dual solution: we set $\alpha_v=1$, $\alpha_{v_i} = \beta_{e_{i+1}} = \frac{1}{2}$ for all $i=1,\dots,k-1$ and $\alpha_{v_k} = \beta_{e_{1}} = \frac{1}{2}$. 
Thus, summing up all these dual values in the dual objective function, we get the valid lower bound result $z_{LP}^{(C)} \geq \Delta(G)+1$.
\qed
\end{proof}

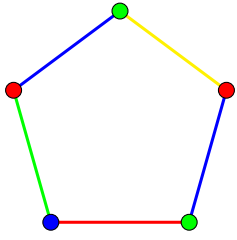
\begin{figure}[t!]
		\centering
		\begin{tikzpicture}[scale=0.35]
    	% first color
	    %second color
		
		\coordinate (v') at (2*1.2-2.5+16,2*1.5-3.2) ;
		\coordinate (w') at (2*3.8-2.5+16,2*1.5-3.2) ;
		\coordinate (z') at (2*4.5-2.5+16,2*4-3.2) ;
		\coordinate (x') at (2*0.5-2.5+16,2*4-3.2) ;
		\coordinate  (y') at (2*2.5-2.5+16,2*5.5-3.2);

	    \draw[green, line width=0.04cm] (x')--(v');
	    \draw[blue, line width=0.04cm] (w')--(z');
		\draw[blue, line width=0.04cm]  (x')--(y');
		\draw[red, line width=0.04cm]  (v')--(w');
		\draw[yellow, line width=0.04cm] (z')--(y');
		
		\draw[fill=blue]  (v') circle [radius=0.30cm];
		\draw[fill=green] (w') circle [radius=0.30cm];
		\draw[fill=red] (z') circle [radius=0.30cm];
		\draw[fill=green] (y') circle [radius=0.30cm];
		\draw[fill=red] (x') circle [radius=0.30cm];
	
		\end{tikzpicture}
		
		\caption{A total coloring of a cycle of length $5$ with $k=4=\Delta(G)+2$ colors. The optimal value of the LP relaxation of the assignment model \eqref{m1:obj}--\eqref{m1:y} is equal to $z_{LP}^{(A)}=3$, while the optimal value of the LP relaxation of the set covering model \eqref{m2:obj}--\eqref{m2:l} is equal to $z_{LP}^{(C)}=\frac{10}{3}$.}
		\label{fig:five odd hole}
\end{figure}

The example in Figure \ref{fig:five odd hole} shows that the optimal value of the LP relaxation of the set covering model can be tighter than the value of the LP assignment relaxation.
Next, we prove that the LP covering relaxation always provides a lower bound at least as strong as that of the LP assignment relaxation. Our proof uses the equivalence of the set covering relaxation $z_{LP}^{(C)}$ with a set partitioning relaxation, where the inequality constraints \eqref{m2:c1}--\eqref{m2:c2} are replaced with equality constraints. Herein, we denote by  $z_{LP}^{(P)}$ the optimal value of the LP partitioning relaxation. The proof of the following result is straightforward.
\begin{lemma}\label{cov-part}
$z_{LP}^{(C)} = z_{LP}^{(P)}$.
\end{lemma}
%
%Observe that the equality constraints in the partitioning model are obtained by applying a
%
%Dantzig-Wolfe reformulation of the constraints \eqref{m1:c4}--\eqref{m1:c5}.
%
We remark that the set partitioning model can be obtained by applying Dantzig-Wolfe decomposition to the assignment formulation, keeping constraints (2) and (3) in the master and moving constraints (4) and (5) to the subproblem.
For a detailed survey of this technique, we refer the reader to \cite{Barnhart1998}.

We are now ready to prove the following proposition.
\begin{proposition}
Let $G=(V,E)$ be a graph. Then, we have $\chi_T(G) \geq z_{LP}^{(C)} \geq z_{LP}^{(A)} \geq \Delta+1$.
\end{proposition}
\begin{proof}
We first show that any feasible solution to the LP relaxation of the partitioning model can be converted into a feasible solution of the LP relaxation of the assignment model with the same objective value function.
Then, by Lemma \ref{cov-part}, we get the result for the optimal value of the set covering model.
We denote by $\lambda^{*}$ a feasible solution to the linear relaxation of the partitioning model.
Now, consider the mapping $\phi:\mathcal{T} \longrightarrow K$ which assigns to each total matching a color.
Notice that, since we are considering maximal total matchings, we can uniquely define a color class for each total matching.
Let us define $(x^{*}_{vk},y^{*}_{ek},z^{*}_{k})$ a feasible solution of the LP relaxation of the assignment model in the following way:
\begin{equation*}
    x_{vk}^{*} = \sum\limits_{t \in \mathcal{T}:v \in t,\phi(t)=k} \lambda^{*}_{t}, \quad y_{ek}^{*} = \sum\limits_{t \in \mathcal{T}:e \in t,\phi(t)=k} \lambda^{*}_{t} \quad\text{and}\quad z_k^{*} = \sum\limits_{t \in \mathcal{T}: \phi(t)=k}\lambda^{*}_{t}.
\end{equation*}
By construction of the solution, we observe that the constraints $x_{vk}^{*}+ \sum\limits_{ e \in \delta(v)}x_{vk}^{*} \leq z^{*}_k$ and $x_{vk}^{*}+x_{wk}^{*}+y_{ek}^{*} \leq z_{k}^{*}$ are satisfied.
Moreover, by the equality constraints imposed in the partitioning model, the assignment constraints \eqref{m1:c1}--\eqref{m1:c3} are satisfied.
Then, we show an instance of graph where $z_{LP}^{(C)}>z_{LP}^{(A)}$.
Consider the cycle of length $5$ of Figure 1.
It turns out that the optimal value of $z_{LP}^{(A)}$ is 3, and it is obtained by setting
$x_{v1}=1, \forall v \in V(C_5)$, $x_{vk} =0, \forall k \in K\setminus\{1\}$, $y_{e1}=1, \forall e \in E(C_5)$ and $z_1 = 3,z_k = 0, \forall k \in K\setminus\{1\}$.
The optimal value $z_{LP}^{(C)}$ is $\frac{10}{3}$, and it uses four total matchings.
It is obtained by setting $\lambda_{t_1}^{*}= \frac{1}{3}$, $\lambda_{t_j}^{*}=1, j =2,3,4$ for the remaining total matchings.
The reader can easily verify that this is an optimal solution of this instance.
This completes the proof.\qed
\end{proof}

\subsection{Column generation and Weighted Total Matchings} 
We can solve problem \eqref{m3:primal}, or equivalently its dual \eqref{m3:d1}--\eqref{m3:d3}, by considering a subset $\mathcal{\bar{T}} \subset \mathcal{T}$, and by applying a Column Generation algorithm, where looking for a primal negative reduced cost variables corresponds to look for a violated dual constraint \cite{Lubbecke2005,Barnhart1998,Desaulniers2006,Gualandi2013}.
Given a dual feasible solution $\bar{\alpha}$ and $\bar{\beta}$, the separation problem of the dual constraints \eqref{m3:d2} reduces to the following Maximum Weighted Total Matching.
%
%\revision{
\begin{align}
\label{m5:obj}   \nu_T(G, \bar \alpha, \bar \beta) := \, \max \quad & \sum_{v \in V} \bar{\alpha}_v x_v + \sum_{e \in E} \bar{\beta}_e y_e \\
\label{m5:c1}    \mbox{s.t.} \quad 
                 & x_v + \sum_{e \in \delta(v)} y_e  \leq 1 & \forall v \in V \\
\label{m5:c2}    & x_v + x_w + y_e\leq 1  & \forall e=\{v,w\} \in E \\
\label{m5:vr}    & x_v,y_e \in \{0,1\} & \forall v \in V, \forall e \in E.
\end{align}

\noindent Note that constraints \eqref{m5:c1} and \eqref{m5:c2} together define the valid constraints for total matchings of $G$. 
In addition, whenever the optimal value $\nu_T(G, \bar \alpha, \bar \beta) > 1$, the corresponding total matching gives a violated constraint \eqref{m3:d2}.
That is, problem \eqref{m5:obj}--\eqref{m5:vr} is the pricing subproblem for solving our set covering model by Column Generation.

Motivated by the solution of the pricing subproblem \eqref{m5:obj}--\eqref{m5:vr}, in the next section, we study valid (facet) inequalities of the Total Matching Polytope.

%%%%%%%%%%%%%%%%%%%%%%%%%%%%%%%%%%%%%%%%%%%%%%%%%%%%%%%%%%%%%
\section{Facet inequalities for the Total Matching Polytope}\label{sec:matching}
In this section, we study the feasible region of the MWTMP \eqref{m5:obj}--\eqref{m5:vr}, and we provide facet-defining inequalities for the corresponding polytope.
The most important original contribution of this paper is given in Theorem \ref{kappaEven}, where we prove that %specific cycle inequalities and 
the even-clique inequalities are facet-defining for the Total Matching Polytope.

\subsection{Total Matching Polytope}
The Total Matching Polytope is defined as the {\it convex hull} of characteristic vectors of total matchings. 
Hence, given a total matching $T$, the corresponding characteristic vector is defined as follows.
\begin{equation*}
  \chi[T]=\left\{
  \begin{array}{@{}ll@{}}
    z_{a}=1 & \text{if}\ a \in T \subseteq D = V \cup E, \\
    z_{a}=0 & \text{otherwise}.
  \end{array}\right.
\end{equation*} 
where $z = (x,y) \in \{0,1\}^{n+m}$, $x$ corresponds to the vertex variables and $y$ to the edges variables.
\begin{definition}
The {\bf Total Matching Polytope} of a graph $G=(V,E)$ is defined as:
\[
P_{T}(G) := \mbox{conv}\{\chi[T] \subseteq \mathbb{R}^{n+m} \mid T \subseteq D = V \cup E \mbox{ is a total matching} \}.
\]
\end{definition}
%\begin{proposition}
%$P_{T}(G)$ has the following valid inequalities:
%
%\begin{align}
%\label{m6:c1} & \sum_{e \in \delta(v)} y_{e} + x_v \leq 1 & \forall v \in V \\
%\label{m6:c2} & x_{v} + x_{w} + y_e \leq 1 & \forall e=\{v,w\} \in E \\
%\label{m6:c3} & x_{v},y_{e} \geq 0 & \forall v \in V, \forall e \in E .
%\end{align}
%\end{proposition}
%\begin{proof}
%In a total matching, by definition, we can take for each vertex $v$ at most one edge incident to $v$ or the vertex %itself (constraints \eqref{m6:c1}). For every edge $e = \{v,w\}$, a total matching contains at most one element among %$e,v,w$ (constraints \eqref{m6:c2}). Clearly, the variables must be nonnegative. \qed
%\end{proof}
The following proposition implies that the valid inequalities that are facet-defining are nonredundant, and, hence, they represent a minimal system defining $P_{T}(G)$.
\begin{proposition}\label{fulldim}
$P_{T}(G)$ is full-dimensional, that is, $\dim(P_{T}(G))=n+m$.
\end{proposition}

\begin{proof}
We have that the origin, the unit vectors $\chi[\{v\}]$ for every $v \in V$ and $\chi[\{e\}]$ for every $e \in E$ belong to $P_{T}(G)$, and clearly they are linearly independent. 
Thus, we have $n+m+1$ affinely independent points. \qed
\end{proof}

We establish now an important connection between total matchings of a graph and the stable sets of the {\it total graph}, defined as in \cite{Total}.
%
%First, we explain how the total graph graph is constructed.
%
Consider the graph $G$ and its corresponding line graph $L(G)$, that is, the graph obtained from $G$ having one vertex for each edge $e \in E(G)$, and where two vertices are linked by an edge if the corresponding edges in $G$ are incident to the same vertex in $G$.
Starting from the line graph, we construct a new graph $H = (V \cup V(L(G)), E(L(G)) \cup E')$, that we will call the \textit{line-full} graph, where $E'$ is the set of edges connecting the vertices of $L(G)$ to vertices of $G$, if and only if $v \in V(L(G))$ is an edge of $G$.
%and a graph $G'$, 
We call a \textit{doubling} of an edge the operation that adds an edge between a pair of vertices.
\begin{definition}
Let $G$ be a graph and $H$ its corresponding line-full graph. The graph $W$ obtained from $H$ applying a doubling of an edge for every pair of vertices $\{v,w\} \in V(H) \setminus V(L(G))$ such that $e = \{v,w\} \in E(G)$, is called the {\bf total graph} of $G$.
\end{definition}
Note that in the line graph, if $|\delta(v)|=l$, then we have a corresponding clique $K_{l}$. In addition, by doubling the edges, we can create triangles in the total graph. 
Hence, as shown in Figure \ref{fig2}, the total graph can be described as the union of cliques $K_{3}$ and general cliques. 
We can prove that total matchings of $G$ correspond to stable sets of its total graph $W$. 
In the following, we denote as $P_{stable}$ the Stable Set Polytope.

\begin{proposition}\label{total graph}
Let $G$ be a graph and $W$ its total graph. Then, $P_{T}(G)=P_{stable}(W)$.
\end{proposition}
\begin{proof}
The characteristic vectors of the stable sets of $W$ correspond to the characteristic vectors of total matchings of $G$, and, hence, the vertices of $P_{stable}(W)$ are the vertices of $P_{T}(G)$.\qed
\end{proof}
\begin{figure}[!t]
		\centering
		\begin{tikzpicture}[scale=0.50]
	    % first color
	
		\coordinate (v) at (2*1.2-2.5,2*1.5-3.2) ;
		\coordinate (w) at (2*3.8-2.5,2*1.5-3.2) ;
		\coordinate (z) at (2*4.5-2.5,2*4-3.2) ;
		\coordinate (x) at (2*0.5-2.5,2*4-3.2) ;
		\coordinate (c) at (2.5,3.2) ;
		\coordinate  (y) at (2*2.5-2.5,2*5.5-3.2);
		\draw[fill=black,line width=0.03 cm] (x) circle [radius=0.15cm];
		\draw[fill=black,line width=0.03 cm] (y) circle [radius=0.15cm];
		\draw[fill=black,line width=0.03 cm] (z) circle [radius=0.15cm];
		\draw[fill=black,line width=0.03 cm] (w) circle [radius=0.15cm];
		\draw[fill=black,line width=0.03 cm] (v) circle [radius=0.15cm];
		\draw[fill=black,line width=0.03 cm]  (c) circle [radius=0.15cm];
	    \draw[line width=0.02 cm] (c)--(z);
	    \draw[line width=0.02 cm] (c)--(w);
	    \draw[line width=0.02 cm] (y)--(c);
	    \draw[line width=0.02 cm] (c)--(v);
	    \draw[line width=0.02 cm] (c)--(x);
	    
	    %second color
		\coordinate (v1) at (1.2+16,1.5) ;
		\coordinate (w1) at (3.8+16,1.5) ;
		\coordinate (z1) at (4.5+16,4) ;
		\coordinate (x1) at (0.5+16,4) ;
		\coordinate (c1) at (2.5+16,3.2) ;
		\coordinate  (y1) at (2.5+16,5.5);
		\coordinate (v') at (2*1.2-2.5+16,2*1.5-3.2) ;
		\coordinate (w') at (2*3.8-2.5+16,2*1.5-3.2) ;
		\coordinate (z') at (2*4.5-2.5+16,2*4-3.2) ;
		\coordinate (x') at (2*0.5-2.5+16,2*4-3.2) ;
		\coordinate  (y') at (2*2.5-2.5+16,2*5.5-3.2);
		\draw[fill=black,line width=0.03 cm] (x1) circle [radius=0.15cm];
		\draw[fill=black,line width=0.03 cm] (y1) circle [radius=0.15cm];
		\draw[fill=black,line width=0.03 cm] (z1) circle [radius=0.15cm];
		\draw[fill=black,line width=0.03 cm] (w1) circle [radius=0.15cm];
		\draw[fill=black,line width=0.03 cm] (v1) circle [radius=0.15cm];
		\draw[fill=black,line width=0.03 cm] (c1) circle [radius=0.15cm];
		\draw[fill=black,line width=0.03 cm]  (v') circle [radius=0.15cm];
		\draw[fill=black,line width=0.03 cm] (w') circle [radius=0.15cm];
		\draw[fill=black,line width=0.03 cm] (z') circle [radius=0.15cm];
		\draw[fill=black,line width=0.03 cm] (y') circle [radius=0.15cm];
		\draw[fill=black,line width=0.03 cm] (x') circle [radius=0.15cm];
	
	    \draw[blue,line width=0.03 cm] (v1)--(w1);
	    \draw[blue,line width=0.03 cm] (v1)--(z1);
	    \draw[blue,line width=0.03 cm] (v1)--(x1);
	    \draw[blue,line width=0.03 cm] (v1)--(y1);
	    \draw[blue,line width=0.03 cm] (v1)--(c1);
	    \draw[blue,line width=0.03 cm] (w1)--(z1);
	    \draw[blue,line width=0.03 cm] (w1)--(x1);
	    \draw[blue,line width=0.03 cm] (w1)--(y1);
	    \draw[blue,line width=0.03 cm] (w1)--(c1);
	    \draw[blue,line width=0.03 cm] (z1)--(x1);
	    \draw[blue,line width=0.03 cm] (z1)--(y1);
	    \draw[blue,line width=0.03 cm] (z1)--(c1);
	    \draw[blue,line width=0.03 cm] (x1)--(c1);
	    \draw[blue,line width=0.03 cm] (x1)--(y1);
	    \draw[blue,line width=0.03 cm] (c1)--(y1);
	    \draw[line width=0.03 cm] (v')--(v1);
	    \draw[line width=0.03 cm] (w')--(w1);
	    \draw[line width=0.03 cm] (x')--(x1);
	    \draw[line width=0.03 cm] (y')--(y1);
	    \draw[line width=0.03 cm] (z')--(z1);
	    \draw[green, line width=0.03 cm] (z') to[out=140, in =70] (c1);
	    \draw[green, line width=0.03 cm] (v') to[out=90, in =190] (c1);
	    \draw[green, line width=0.03 cm] (x') to[out=50, in =90] (c1);
	    \draw[green, line width=0.03 cm] (y') to[out=-30, in =0] (c1);
	    \draw[green, line width=0.03 cm] (w') to[out=70, in =0] (c1);
		\end{tikzpicture}
		\caption{A star graph on the left and the corresponding total graph on the right.\label{fig2}}
		\label{fig:unified}
\end{figure}
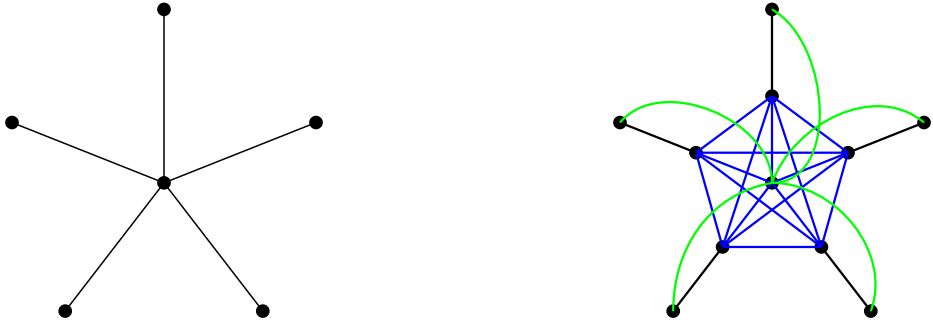
%%%%%%%%%%%%%%%%%%%%%%%%%%%%%%%%%%%%%%%%%%%%%%%%
\subsection{Facet-defining inequalities}
In the following paragraphs, we prove that the valid inequalities describing the feasible region of the Total Matching Polytope are all facet-defining.

\begin{proposition}
Let $G$ be a graph. Then, the inequalities
\begin{align}
\label{b1:v} & x_v + \sum_{e \in \delta(v)} y_{e} \leq 1 & \forall v \in V \\
\label{b1:e} & x_{v} + x_{w} + y_e \leq 1  & \forall e=\{v,w\} \in E \\
\label{b1:a} & x_{v},y_{e} \geq 0 & \forall v \in V, \forall e \in E .
\end{align}
are facet-defining for the Total Matching Polytope $P_{T}(G)$.
\end{proposition}

\begin{proof}
Let $A=(V',E')$ be the total graph associated to $G$.
Now, consider a vertex $v \in V$ and an edge $e:=\{u,w\} \in E$.
By construction of $A$, the subgraphs $A[\delta(v) \cup \{v\}]$ and $A[e \cup \{u,w\}]$ in $A$ correspond to cliques $K_{|\delta(v)|+1}$ and $K_3$ respectively.
Moreover, it is easy to see that they are maximal cliques.
Then, since $P_{stable}(A)= P_{T}(G)$ by Proposition (\ref{total graph}) and using the fact that maximal cliques and nonnegativity constraints are facet-defining inequalities for $P_{stable}(A)$ (e.g., see \cite{Padberg1973}), we get the result.\qed
\end{proof}

\subsection{Perfect Total Matchings}
A total matching is {\it perfect} if every vertex of the graph is covered by a total matching, that is, every vertex is either in the total matching or one of its incident edges belongs to the total matching. 
We prove next, that for any graph $G$, we can always find a perfect total matching.
\begin{proposition}
Every graph $G$ has a perfect total matching.
\end{proposition}
\begin{proof}
If $G$ has a perfect matching, it is trivial.
Otherwise, let us suppose that $G$ has no perfect matching.
Given a subset of vertices $S\subseteq V$, let $k$ be the number of odd components of $G$
%we denote by $\mbox{oddcc}(G \setminus S)$ the number of odd components of $G$, 
that is, the number of maximal connected components of odd order.
We denote the odd components as $O_1,O_2, \dots, O_k$.
By applying the Tutte's theorem \cite{Tutte1954,Lovasz2009}, we have $k > |S|$.
Notice that, since the maximum size of a matching in an odd component is $\frac{|V(O_i)|-1}{2}$ for $i=1,2, \dots k$,
there is a vertex that is not covered by a matching, we call it a left-out vertex.
Instead, we have a perfect matching $N$ that covers all the vertices in the even components.
Now, let $T$ be a total matching of $G$.
We show how to construct $T$ so that all the vertices of $G$ are covered by $T$.
First, for each odd component we can construct a maximum matching $M_{i}$ of size $\frac{|V(O_i)|-1}{2}$ for every $i=1,2,\dots k$,
in which we choose as a left-out vertex one of the vertices connecting an odd component to $S$. 
Let $v_{i}$ be the left-out vertex by $M_{i}$ of the component $O_{i}$ for $i=1,2, \dots, k$.
Now, take one edge of $|S|$ odd components connecting $v_{i}$ to the set $S$
and consider the matching $S_{O}:=\{e = \{v_{i},s_{i}\} \mid s_{i} \in V(S)\,$ for $i=1,2, \dots, |S|\}$.
Since $k > |S|$, for each of the remaining components, we have a left-out vertex that cannot be covered by a matching $M_{i}$ and in particular, in order to form an independent set of elements, we cannot choose an edge connecting $S$ to the odd component.
Consider the set $L$ of these vertices and define $T := M_{1} \cup M_2 \dots \cup M_{k} \cup S_{O} \cup L \cup N$.
Since every vertex is covered by $T$ by construction, the assertion follows.
\qed
\end{proof}

The previous proposition allows us to define the Perfect Total Matching Polytope. Let $P_{PT}(G)$ be the convex hull of all perfect total matchings of $G$.
\begin{proposition}
Let $G$ be a graph. The following inequalities are valid for $P_{PT}(G)$.
\begin{align}
&x_{v} + \sum_{e \in \delta(v)} y_{e} = 1  &  \forall v \in V \\
 &x_{v} + x_{w} + y_{e} = 1   &\forall e=\{v,w\} \in E \\
 & x_{v},y_{e} \geq 0 & \forall v \in V, \forall e \in E.
\end{align}
\end{proposition}
\noindent In practice, for any perfect total matching, the inequalities describing the feasible region of total matchings are all tight.
%%%%%%%%%%%%%%%%%%%%%%%%%%%%%%%%%%%%%%%%%%
\begin{figure}[!t]
		\centering
		\begin{tikzpicture}[scale=0.40]
	    % first color
	    %second color
		
		\coordinate (v') at (2*1.2-2.5+16,2*1.5-3.2) ;
		\coordinate (w') at (2*3.8-2.5+16,2*1.5-3.2) ;
		\coordinate (z') at (2*4.5-2.5+16,2*4-3.2) ;
		\coordinate (x') at (2*0.5-2.5+16,2*4-3.2) ;
		\coordinate  (y') at (2*2.5-2.5+16,2*5.5-3.2);
	     
	    \node at (2*1.2-2.5+16,2*1.5-3.2-1){$\frac{1}{3}$};
	    \node at (2*3.8-2.5+16,2*1.5-3.2-1){$\frac{1}{3}$};
	    \node at (2*4.5-2.5+16+1,2*4-3.2){$\frac{1}{3}$};
	    \node at (2*0.5-2.5+16-1,2*4-3.2){$\frac{1}{3}$};
        \node at (2*2.5-2.5+16,2*5.5-3.2+1){$\frac{1}{3}$};
        \node at (18.5,-1.2){$\frac{1}{3}$};
        \node at (18.5-3.8,2.5){$\frac{1}{3}$};
        \node at (18.5+3.8,2.5){$\frac{1}{3}$};
        \node at (18.5-2.1,7){$\frac{1}{3}$};
        \node at (18.5+2.2,7){$\frac{1}{3}$};
        
	    \draw[black,line width=0.03 cm] (x')--(v');
	    \draw[black,line width=0.03 cm] (w')--(z');
		\draw[black,line width=0.03 cm]  (x')--(y');
		\draw[black,line width=0.03 cm]  (v')--(w');
		\draw[black,line width=0.03 cm] (z')--(y');
		
		\draw[fill=black]  (v') circle [radius=0.30cm];
		\draw[fill=black] (w') circle [radius=0.30cm];
		\draw[fill=black] (z') circle [radius=0.30cm];
		\draw[fill=black] (y') circle [radius=0.30cm];
		\draw[fill=black] (x') circle [radius=0.30cm];
	
		\end{tikzpicture}
		
		\caption{A vertex $\bm{z}=\frac{1}{3}\mathbf{1}$ of the cycle $C_5$. \label{fig3}}
\end{figure}
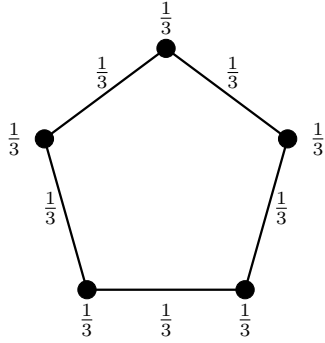
In the following section, we introduce nontrivial facet-defining inequalities for the Total Matching Polytope.
\section{Clique inequalities and Congruent-$2k3$ cycle inequalities}\label{sec:valid}  %even-clique inequalities
In the previous section, we have proved that all the inequalities defining the feasible region of the TMP are facet-defining. 
In this section, we introduce three families of nontrivial valid inequalities.
%%%%%%%%%%%%%%%%%%%%%%%%%%%%%%%%%%%%%%%%%%%%%%%%%%%%%%
First, notice that the result obtained by Padberg in \cite{Padberg1973} for a maximal clique on the intersection graph of a set packing problem can be extended to the Total Matching Polytope,
since the total graph can be interpreted as a special type of intersection graph.
In particular, Padberg shows that any maximal clique on the intersection graph induces a facet-defining inequality for the packing polytope.
Hence, any maximal clique on the total graph induces a facet-defining inequality for the Total Matching Polytope.
However, in the following, we propose direct proofs using only the original graph $G$ to show that the families of valid inequalities that we propose are facet-defining.
\begin{theorem}\label{PadbergMaxClique}
Let $G$ be a graph and let $K_h$ be a maximal clique of $G$. Then, the vertex-clique inequality
\begin{align}\label{vertex-clique}
    \sum_{v \in V(K_h)} x_v \leq 1
\end{align}
is facet-defining for $P_{T}(G)$.
\end{theorem}
\begin{proof}
Let $G$ be a graph and let $K_h \subseteq G$ be a maximal clique.
We have to exhibit $n+m$ affinely independent points which belong to the face $F:=\{z \in P_{T}(G) \mid \sum\limits_{v \in V(K_h)}x_v = 1\}$. %induced by the maximal clique.
We know that, since $K_h$ is maximal, by Theorem 2.4 in \cite{Padberg1973}, we can easily construct $n$ of such points belonging to $F$.
Now, fix a vertex $v \in V(K_h)$ and denote $N_{G}(v) \cap V(K_h):=\{v_0,v_1, \dots, v_{h-2}\}$ and $\delta(v) \cap E(K_h): = \{e_{i}=\{v, v_{i}\} \mid i=0,1, \dots, h-2\}$.
Define the total matching $T_{v}^{\overline{e}}:=\{v,\overline{e}\}$, where $\overline{e}\notin \delta(v)$
and notice that $\chi[T_{v_i}^{\overline{e}}] \in F$.
Thus, we can construct the set of 0-1 vectors $\{ \chi [T_{v}^{\overline{e}}] \mid $ $ \forall e \notin \delta(v) \} \subseteq F$.
It is easy to see that the corresponding characteristic vectors are affinely independent, so up to now we have found $|E \setminus \delta(v)|$.
%
%Then, consider the sets of total matchings $\overline{T_{v}^{w}}:= \{ \{w,e\} \mid w \in V(K_h), \forall e \in \delta(v), e \notin E(K_h)\}$.
%
Then, fix a vertex $w \in V(K_h)$ with $w \neq v$, and consider the set of total matchings $T_w^{e}:=\{\{w,e\} \mid \forall e \in \delta(V(K_h)) \cap \delta(v)$\}.
By construction, the set of the corresponding characteristic vectors of $T_w^{e}$ is contained in $F$.
Finally, the vectors $\chi[T_{v_i,e_{i+1}}]$ for $i=0,1,\dots, h-2 \mod h-1$ with one in entry $x_{v_i}$ and $y_{e_{i+1}}$ and zero the other components, are characteristic vectors lying on $F$ and are affinely independent, thus we have $|E|$ affinely independent points.
We have found in total $n+m$ affinely independent points, since the matrix having the columns the characteristic vectors found assumes the following form:
\begin{center}
$\left[
\begin{array}{c|c}
A_{v} & B \\ \hline
\mathbf{0} & I_{e} \\ 
\end{array}\right],
$
\end{center}
where $A_v$ represents the vertex components of the $n$ points and $B$ the characteristic vectors of total matchings relative to a fixed vertex in the clique and exactly one edge not belonging to the clique itself.
This completes the proof. \qed

\end{proof}

%%%%%%%%%%%%%%%%%%%%%%%%%%%%%%%%%%%%%%%%%%%%%%%%%%%%%%
\subsection{Congruent-$2k3$ cycle inequalities}
The inequalities \eqref{m5:c1}--\eqref{m5:vr} define the feasible region of total matchings and they are facet-defining, but they do not describe the complete convex hull of the Total Matching Polytope. For instance, Figure \ref{fig3} shows that using only those inequalities, we have that for a cycle $C$ of length 5, the point $z_a = \frac{1}{3}$ for all $a \in V(C) \cup E(C)$ belongs to $P_T(C)$ and it is a vertex.
However, in \cite{Leidner2012}, the authors show that the cardinality of a maximum total matching in a cycle of cardinality $k \in \mathbb{N}$ is equal to $\floor{\frac{2k}{3}}$. 
Thus, we introduce an inequality that cuts off these nonintegral solutions for cycles, which we call the {\bf congruent-$2k3$ cycle inequality}.

\begin{proposition}\label{mainresult}
Let $C_k$ be an induced cycle. Then, if 
%For every induced cycle $C_{k} \subseteq G$, with 
$k \equiv 1 \mod 3$ or $ k \equiv 2 \mod 3$, the {\bf congruent-$2k3$ cycle inequality} defined as
\begin{equation}\label{cycle}
    \sum\limits_{v \in V(C_k)}x_{v}+\sum\limits_{e \in E(C_k)}y_{e} \leq \left\lfloor \frac{2k}{3} \right\rfloor
\end{equation}
is facet-defining for $P_{T}(C_k)$.
\end{proposition}
\begin{proof}
Let $F :=\{z \in P_{T}(C_k) \mid \lambda^{T}z=\lambda_{0} \}$ be a facet of $P_{T}(C_k)$ such that $\tilde{F} := \{z \in P_{T}(C_k) \mid \tilde{\lambda}^{T}z=\tilde{\lambda_{0}} \} \subseteq F$ where the inequality $\tilde{\lambda}^{T}z \leq \tilde{\lambda_{0}}$ corresponds to the inequality \eqref{cycle}. 
We want to prove that there exists $a \in \mathbb{R}$ such that $\lambda = a\tilde{\lambda}$ and $ \lambda_{0}=a\tilde{\lambda_{0}}$. 
We distinguish two cases based on the parity of the cycle. 
We label the vertices $V(C_k):=\{v_{0},\dots,v_{k-1}\}$, so that $v_{i}$ is adjacent to $v_{i-1}$ for $i=0,1,\dots,k-1 \mod k$ %, and $v_{0}$ is adjacent to $v_{k-1}$. 
%We label 
, and the edges $E(C_k):=\{e_{0},\dots,e_{k-1}\}$, so that $e_{i} = \{v_{i}, v_{i+1}\}$ for $i=0,1,\dots,k-1 \mod k$. % and the last edge is $e_{k}= \{v_{k}, v_{1}\}$.

\textbf{Case 1: $(k \equiv 1 \mod 3)$}. Consider the total matching $T_{0} := \{v_{i}, e_{i+1} \mid 0 \leq i \leq k-4,$ for $i \equiv 0 \mod 3 \}$. 
This is a maximal total matching, since every element in $T_0$ is mutually nonadjacent.
The number of elements of $T_{0}$ is twice the numbers of integers $i$ satisfying the condition, that is, $|T_0| = \frac{2(k-1)}{3}$, and, hence, $\chi[T_{0}] \in \tilde{F}$ and, in particular, $\chi[T_{0}] \in F$.
%
%This implies that $\lambda_{v} = \lambda_{e} = 0$, for $ v,e \notin T_{1}$ ,
%
%where $\lambda_{v}$ is the cost coefficient corresponding to the vertex $v \in V$ and $\lambda_{e}$ is the coefficient relative to the edge $e \in E$.
%
Note that the set $\{v_{k-2},e_{k-2}\}$ is not contained in $T_0$, because of our description of $T_{0}$.
%$k \equiv 1 \mod 3$. 
%
Now, consider the total matchings $T_{0}^{-}:=(T_{0} \setminus \{e_{k-3}\}) \cup \{v_{k-2}\}$ and $T_{0}^{+}:=(T_{0} \setminus \{e_{k-3}\}) \cup \{e_{k-2}\}$.
In this way, we obtain two distinct total matchings with the same cardinality, whose characteristic vectors belong to $\tilde{F}$. 
%
%This implies that $\lambda_{v} =\lambda_{e} = 0$ for $v,e \notin T_{1}^{+}$.
%We have that $\chi[T_{1}^{+}]\in \tilde{F}$ and $\chi[T_{1}^{-}] \in \tilde{F}$. 
%
Since $\chi[T_{0}^{+}]\in F$ and $\chi[T_{0}^{-}] \in F$, then $\lambda^{T}\chi[T_{0}]= \lambda^{T}\chi[T_{0}^{+}]$ and $\lambda^{T}\chi[T_{0}]=\lambda^{T}\chi[T_{0}^{-}]$, thus $\lambda_{e_{k-3}}=\lambda_{v_{k-2}}=\lambda_{e_{k-2}}$, %= \lambda_{0}$,
where $\lambda_{v_{i}}$ is the cost coefficient corresponding to the vertex $v_{i}$ and $\lambda_{e_{i}}$ is the coefficient relative to the edge $e_i=\{v_{i},v_{i+1}\}$.
Now, consider the function $\sigma:C \longrightarrow C$ such that $\sigma(v_{i})=v_{i+1}$ and $\sigma(e_{i})=e_{i+1}$.
Indeed, $\sigma$ shifts every element to the next position with respect to the ordering of the vertices and the edges. 
Composing $k-1$ times the shifting function on $T_{0}$, we obtain the following total matchings $\sigma(T_{0}),\sigma^{2}(T_{0}),\dots,\sigma^{k-1}(T_{0})$.
For a fixed $i$, denote $\sigma^{i}(T_{0}):=T_{i}$. 
These are still total matchings and each characteristic vector $\chi[T_{i}]  \in \tilde{F}$. 
%thus $\lambda_{v} =\lambda_{e} = 0$, for $v,e \notin T_{i}$. 
%
Notice also that $T_{i}$ does not contain $\{v_{i-2},e_{i-2}\}$, for $i=1$ the corresponding set is $\{v_{k-1},e_{k-1}\}$. 
So, following the same previous procedure, we deduce that $\lambda_{e_{i-3}}=\lambda_{v_{i-2}}=\lambda_{e_{i-2}}$% = \lambda_{0}$ 
for all $i =1,\dots,k-1$ mod $k$. 
This implies that there exists $a \in \mathbb{R}$ such that $\lambda = a\mathbf{1}$.
Then, since $\chi[T_i] \in \tilde{F}$, we have that $\lambda^{T}\chi[T_i] = a(\mathbf{1}^{T}\chi[T_i])=a\tilde{\lambda_0} $.
We conclude that, since $(\lambda,\lambda_0)= a(\mathbf{1},\tilde{\lambda_{0}})$,
$\lambda^{T}z  \leq \lambda_{0}$ is a scalar multiple of the cycle inequality.
\\
\textbf{Case 2: $(k \equiv 2 \mod 3$)}. Consider the total matching $T_{0} :=\{v_{i},e_{i+1}\mid 0 \leq i \leq k-5,$ for $i \equiv 0 \mod 3 \} \cup \{v_{k-2}\}$. 
Notice that now $e_{k-2} \notin T_{0}$.
Also in this case $\chi[T_{0}] \in \tilde{F}$, since $|T_{0}|=\frac{2(k-2)}{3}+1=\floor{\frac{2k}{3}}$.
We can construct other two total matchings with the same cardinality $\hat{T_{0}} :=(T_{0}\setminus \{v_{k-2}\}) \cup \{e_{k-2}\}$ and $\tilde{T_{0}} :=(\hat{T_{0}} \setminus \{e_{k-4}\}) \cup \{v_{k-3}\}$.
Note that $\chi[\hat{T_{0}}],\chi[\tilde{T_{0}}] \in \tilde{F}$, and so they also lie in $F$.
Thus, $\lambda^{T}\chi[\hat{T_{0}}]=\lambda^{T}\chi[T_{0}]$ and $\lambda^{T}\chi[\hat{T_{0}}]=\lambda^{T}\chi[\tilde{T_{0}}]$.
From the first equality, we deduce that $\lambda_{v_{k-2}}=\lambda_{e_{k-2}}$ and for the second one, $\lambda_{v_{k-3}}=\lambda_{e_{k-4}}$. 
We conclude as in the Case 1 by applying the shifting function $\sigma$, so we have a scalar multiple of the cycle inequality.\qed
\end{proof}
A different proof of this result was first given by Trotter in \cite{antiweb}\footnote{We thanks the anonymous reviewer for pointing this out}. 
Let $p$ and $q$ be two integers with $p \geq 2q+1$. 
Trotter introduced the web inequalities $\sum\limits_{i \in W} x_i \leq q$, where $W \subseteq V$ induces a web $W(p, q)$ of $G$,
that is, a subgraph with $p$ vertices $\{1, \dots , p\}$ with, adopting modulo $p$ arithmetic, an edge between every two
vertices $i$ and $j \in \{i+q, \dots, i-q\}$. 
In the same paper \cite{antiweb}, Trotter introduced the antiweb inequalities $\sum\limits_{i \in \overline{W}} x_i \leq \floor{\frac{p}{q}}$, 
where $\overline{W}(p,q)$ is the complement of $W(p,q)$, that is, a web $W(p,q)$ of $\overline{G}$.
By construction, it is possible to notice that the total graph $T(C_k)$ of a congruent-$2k3$ cycle is an antiweb $\overline{W}(p,3)$,
where $p = 2k$ with $k \in \mathbb{N}$, and $q = 3$.
If $G$ is itself an antiweb, and if $p$ and $q$ are relatively prime, then the antiweb inequalities are facet-defining.
In our case, $p = 2k$ and $q=3$, thus $gcd(p,q)=1$ if and only if $k \equiv 1,2 \mod 3$.
We give a direct proof of Proposition (\ref{mainresult}), without using the total graph, 
since in the total graph we have a loss in structure, in the sense that we cannot any longer distinguish among vertices and edges of the original graph $G$.
Using another direct proof, we proceed in proving the following important observation.

%
%In particular, we notice that when the induced cycle is $C_4$, 
%
%then the corresponding cycle inequality is facet-defining for the Total Matching Polytope of the entire graph.
%
\begin{proposition}\label{cycle4}
Let $G$ be a graph and let $C_4$ be the induced cycle of four vertices. Then, the inequality:
\begin{align}\label{kappaFour}
    \sum\limits_{v \in V(C_4)}x_v + \sum\limits_{e \in E(C_4)}y_e \leq 2
\end{align}
is facet-defining for $P_{T}(G)$.
\end{proposition}
\begin{proof}
Denote by $\tilde{F}$ the face induced by the inequality \eqref{kappaFour}.
Suppose by contradiction that $\Tilde{F}$ is contained in $ F:=\{ z \in P_{T}(G) \mid \lambda^{T}z = \lambda_0\}$.
By proposition \eqref{mainresult}, the corresponding inequality inducing the face $F$ has the form $a(\sum\limits_{v \in V(C_4)}x_v + \sum\limits_{e \in E(C_4)} y_e) + \sum\limits_{l \notin C_4}\lambda_{l}^{T}z_l \leq 2a$ for $a \in \mathbb{R}$.
Denote as $V(C_4):=\{v_0,v_1,v_2,v_3\}$ and $E(C_4):=\{e_{i,i+1}=\{v_i,v_{i+1}\} \mid i=0,1,2,3 \mod 4\}$.
Consider the matching $M:=\{e_{0,1},e_{2,3}\}$, then the corresponding characteristic vector lies on $\tilde{F}.$
Since $M \cap \{u\} = \emptyset$ for every $u \notin V(C_4)$, $T_u:= M \cup \{u\}$ is a total matching whose characteristic vector lies on $\tilde{F}$.
This implies that $\lambda_{u}=0$ for every $u \notin V(C_4)$.
Similarly, $M \cap \{e\} = \emptyset$ for every $e \notin \delta(V(C_4)) \cup E(C_4)$, so $T_e:= M \cup \{e\}$ is a total matching whose characteristic vector lies on $\tilde{F}$.
This implies that $\lambda_{e}=0$ for every $e \notin \delta(V(C_4)) \cup E(C_4)$.
Now, let $S:=\{v \in V(C_4) \mid \delta(V(C_4)) \neq \emptyset \}$.
Fix a vertex $v_i \in S$, and consider the total matching $T_{v_i}:= \{e_{i+1,i+2},v_{i+3}\} \cup \{e_{v_{i}}\}$, where $i$ is taken modulo $4$ and $e_{v_{i}} \in \delta(V(K_h)) \cap \delta(v_i)$.
It is easy to see that the characteristic vector of $T_{v_i}$ lies on $\tilde{F}$, in particular exactly one edge $e \in \delta(v_i) \cap \delta(V(C_4))$ is chosen, so $\lambda_{e_{v_i}}=0$ for every $e_{v_{i}} \in \delta(v_i) \cap \delta(V(C_4))$.
In this way, repeating the same argument for all $v \in S$, we obtain that $\lambda_{e_{v_k}} =0$ for every $e_{v_k} \in \delta(V(K_h))$.
This completes the proof since we have proved that $\lambda_{l}=0$ for all $l \notin C_4$. \qed

\end{proof}
%%%%%%%%%%%%%%%%%%%%%%%%%%%%%%%%%%%%%%%%%%%%%%%%%%%%%%%%%%%%%%%%%%
%\newpage
\paragraph{Separation of congruent-2$k$3 cycle inequalities} 
In this paragraph, we deal with the problem of separating the facet-defining inequalities given by the class of the congruent-$2k3$ cycle inequalities.
Given a fractional optimal solution of the LP relaxation of the pricing subproblem, the separation for the congruent-$2k3$ cycle inequalities consists of either finding an inequality in this class that is violated by a cycle inequality or proving that all inequalities are satisfied.
To this end, we propose an Integer Linear Programming formulation for solving this separation problem.

Let $(c_v,w_e)$ denote the fractional optimal solution to the current LP problem, and let $x_v$ and $y_e$ denote the decision variables of the problem of finding a congruent-$2k3$ cycle in a graph $G$.
The separation problem consists of maximizing the following value
\begin{equation}\label{eq:floor}
 \alpha := \sum\limits_{v \in V}c_v x_v + \sum\limits_{e \in E}w_e y_e - \floor{\frac{2k}{3}},
\end{equation}
\noindent where $k$ is the cardinality of the cycle induced by the variables $x_v$ and $y_e$.
Thus, we want to detect a maximum weighted cycle, where node and edge weights are  $(c_v,w_e)$, and the cycle contains a number of nodes that is not a multiple of three. 
Whenever $\alpha > 0$, we have a violated cycle. 
Otherwise all the congruent-$2k3$ cycle inequalities are satisfied.
Since $k \equiv 1,2 \mod 3$, we can express $k = 3z + t$ where $z \in \mathbb{Z}$ and $t \in \{1,2\}$, and we can rewrite the floor expression in \eqref{eq:floor} as follows
\begin{equation*}
    \floor{\frac{2k}{3}} = \floor{\frac{2(3z+t)}{3}} =
    \left\{
    \begin{array}{ll}
        2z & \mbox{ if } t=1 \\
        2z + 1 &  \mbox{ if } t=2,
    \end{array}\right.
\end{equation*}
\noindent and, hence, we get
\begin{equation}\label{floor-rewritten}
    \floor{\frac{2k}{3}} = 2z + t - 1.
\end{equation}
Another important element of our ILP model for the separation of congruent-$2k3$ cycle inequalities is the connectivity constraints, which we formulate exploiting the ideas presented in \cite{Maksimovic}, by setting a network flow model.
Given the original graph $G=(V,E)$ the flow networks is defined as $H=(V,A)$, where $A:=\bigcup_{\{i,j\}\in E}\{(i,j), (j,i)\}$.
The network $H$ has a single source node that introduces all the flow, while every node that belongs to the cycle is the sink of a single unit of flow.
However, we do not fix in advance the source node, and we let variables $s_i \in \{0,1\}$ for $i = 1, \dots, n $ to indicate which node of $H$ is the source.
Then, we introduce the variables $u_i \in \mathbb{Z}_+$ for every vertex $v_i \in V$ to indicate the overall amount of flow originated at the only source node $i$ having $s_i=1$. Indeed, we have that $u_i > 0$ only for the sink node.
The complete ILP model for the separation of congruent-$2k3$ cycle inequalities is the following:
\begin{align}
\label{m3:obj}  \max \quad & \sum_{v \in V}c_v x_v +  \sum_{e \in E} w_e y_e  - (2z + t - 1) & \\
\label{m3:c1}    \mbox{s.t.} \quad 
    & \sum_{e \in \delta(v)} y_{e} = 2x_v & \forall v \in V \\
\label{m3:c3}   & \sum_{v \in V} x_v =3z+t & \forall v \in V,\forall e \in E \\
\label{m3:c9} & 
x_{i} + \sum_{(i,j) \in A} f_{ij} = u_{i} + \sum_{(j,i) \in A} f_{ji}  
& \forall i \in V \\
\label{m3:c6} & \sum\limits_{i=1}^{n} s_i = 1 & \forall i \in V \\
\label{m3:c7} & u_{i} \leq n \cdot s_{i} & \forall i \in V \\
\label{m3:c8} &  f_{ij} \leq n \cdot y_{e} & \forall (i,j) \in A \\
\label{m3:c10}  & y_e,x_v \in \{0,1\} & \forall v \in V, \forall e \in E \\
\label{m3:c14} & u_i \in \mathbb{Z}_+ & \forall i \in V \\
\label{m3:c11} & z \in \mathbb{Z}_+, t \in \{1,2\}.
\end{align}
The objective function \eqref{m3:obj} includes the relation specified in \eqref{floor-rewritten}.
Constraints \eqref{m3:c1} ensure that the subgraph induced by the variables $x_v$ and $y_e$ is a union of disjoint cycles, since every node has either degree zero or two.
Constraints \eqref{m3:c3} impose the congruence on the length of the cycle, which cannot be a multiple of three.
Constraints \eqref{m3:c9} impose the flow conservation at every node, and constraints \eqref{m3:c6} impose that a single vertex is the origin of the flow.
Constraints \eqref{m3:c7} impose that all the vertices but the source have $u_i=0$, that is, they do not originate any unit of flow.
For every flow variable $f_{ij}$, constraints \eqref{m3:c8} set the capacity of the flow variables to zero whenever $y_e=0$, that is, whenever arc $e$ is not included in the cycle. 

%%%%%%%%%%%%%%%%%%%%%%%%%%%%%%%%%%%%%%%%%%%%%%%%%%%%%%%%%%%%%
The ILP model \eqref{m3:obj}--\eqref{m3:c11} permits us to look for the most violated congruent-$2k3$ cycle inequality by solving a single problem. 
Alternatively, we could solve a simplified version of the separation problem by fixing in advance both the source node $s_i$ and the value of variable $t$.
In this way, to find the most violated inequality, we have to solve two (easier) subproblems for every node, for a total of $6n$ subproblems.
However, each subproblem reduces to a Shortest Path Problem defined on an auxiliary directed graph having nonnegative weights, as shown in the proof of the following proposition.
\begin{proposition}
The separation problem of the congruent-$2k3$ cycle inequality is in $P$.
\end{proposition}
\begin{proof}
The separation problem consists of a sequence of $2n$ Minimum Weighted $s,t$-Path problems from a source node $s$ to the target node $t$ of an auxiliary graph. 
Let $G=(V,E)$ be a weighted graph where $(c_v,w_e)$ are the optimal values of the current LP relaxation.
Starting from $G=(V,E)$, we construct a weighted directed graph $H=(N,A)$ in the following way.
For every vertex $v \in V$, we introduce three nodes labelled as $v_0,v_1,v_2$ in $N$.
Now, for each edge $e=\{v,w\} \in E$, we introduce three arcs $a_i \in A$ 
%between two triples of nodes 
with respect to the permutation $\sigma = (0 1 2)$, that is, 
$a_i = (v_{i},w_{\sigma(i)})$, with $i=0,1,2$. 
%$A_{v,w} = \{ (v_{i},w_{\sigma(i)})\} \mid i=0,1,2\}$ 
%is the set of arcs from $v$ to $w$ in $H$.
%
%Following the scheme given by permutation $\sigma$, every time we switch in a different level (color) from one triple of nodes to another.
%
Observe that a path from $v_0$ to $v_1$ gives a path $P_k$ of size $k \equiv 1\mod 3$, and 
a path from $v_0$ to $v_2$ gives a path $P_k$ of size $k \equiv 2\mod 3$.
Next, we distinguish the two cases.
%where $i=1,2$ gives a path $P_k$ of size $k \equiv 1\mod 3$ if $i =1$, and $k \equiv 2 \mod 3$ if $i=2$, respectively.
%
%Fix the length of a path $P$ to find, thus let $k = 3z+t$, where $t \in \{1,2\}$ is fixed.
%
%Based on the parity of the path to find we define the costs on the arcs.
%
\paragraph{{\bf Case 1:} $(k \equiv 1 \mod 3)$} 
In this case, we have $\floor{\frac{2k}{3}}=\frac{2(k-1)}{3}$, and the separation problem reads as follows:
\begin{align*}
     \exists C_k \ :\ \frac{2}{3}|C_k|-\sum\limits_{v \in V(C_k)}{c_vx_v}-\sum\limits_{e \in E(C_k)}{w_ex_e} < \frac{2}{3}.
\end{align*}
Since we look for the most violated inequality, for each node $v\in V$, the separation problem is equivalent to a Minimum Weighted $s,t$-Path Problem where the source is $v_0$ and the target is $v_1$.
Now, we define the costs on the arcs as $l_{a=(i,j)}:=\frac{2}{3} - c_i -w_{e=\{i,j\}} + 1$, for every $a=(i,j) \in A$. 
We know that $c_i + c_j + w_{e=\{i,j\}} \leq 1 $ due to feasibility of constraints \eqref{m5:c2} and, hence, the costs are positive.
Let $P_1$ be a minimum weighted path in $H$ from $v_0$ to $v_1$.
By construction, the path $P_1$ in $H$ corresponds to a cycle $C_k$ in $G$ of length $k \equiv 1 \mod 3$, where for each node $v_i \in N$ we consider the corresponding node $v \in V$.
If we sum up all the costs on the path $P_1$, we obtain:
\begin{equation}
    l(P_1) := \sum_{(i,j) \in A(P_1)} l_{(i,j)} 
    = \frac{2}{3}|P_1| - \sum_{i \in V(C_k)}{c_{i}-\sum_{e \in E(C_k)} w_e } + |P_1|.
\end{equation}
Hence, the path $P_1$ yields a violated congruent-$2k3$ cycle $C_{k}$ in $G$ if and only if $l(P_1) - |P_1| <\frac{2}{3}$.
%Hence, a minimum weighted path in $H$ yields a congruent-$2k3$ cycle $C_{k}$ in $G$ that maximizes the quantity $\sum\limits_{v \in V(C_k)}c_v x_v +  \sum\limits_{e \in E(C_k)} w_e y_e  - (3k+1) $ if and only if $l(P)<\frac{2}{3}$.

\paragraph{{\bf Case 2:} $(k \equiv 2 \mod 3)$} 
In this case, we have $\floor{\frac{2k}{3}}=\frac{2(k-2)}{3}$, and the separation problem reads as follows:
%Similarly, fix $k \equiv 2 \mod 3$, then $\floor{\frac{2k}{3}} = \frac{2(k-2)}{3}$.
%
%Now the separation problem reads as follows:
\begin{align*}
     \exists C_k \ :\ \frac{2}{3}|C_k|-\sum\limits_{v \in V(C_k)}{c_vx_v}-\sum\limits_{e \in E(C_k)}{w_ex_e} < \frac{4}{3}.
\end{align*}
Hence, we have to find a minimum weighted path $P_2$ from $v_0$ to $v_2$ for each node $v$ in $V$.
We define the arc costs $l_a$ as before, and we get a maximum violated cycle 
%i.e. $\sum\limits_{v \in V(C_k)}c_vx_v + \sum\limits_{e \in E(C_k)}w_ey_e - (3k+2)$ 
if and only if $l(P_2) - |P_2| < \frac{4}{3}$.

In conclusion, by solving $2n$ shortest path problems on a directed graph with positive weights, we get the most violated congruent-$2k3$ cycle inequalities in polynomial time.
\qed
\end{proof}
We recall the fact that the separation of antiweb inequalities is in $P$ was proved in \cite{Vries}.
In our case we have a generalization of the odd cycle inequalities by paths in tripartite graphs, instead of bipartite graphs, see \cite{Nemhauser1991}.

%\end{comment}
%%%%%%%%%%%%%%%%%%%%%%%%%%%%%%%%%%%%%%%%%%%%%%%%%%%%%%%%%%%%%%%%%%%%%%%%%%%%%%%%%%%%%%%%%%%%%%%%%%%%%%%%%%%%%%%%%%%%%

\subsection{Even and odd clique inequalities}
At this point, we focus on valid inequalities that can be derived by complete subgraphs $K_h$ of $G$, with $h \leq n$. 
This leads to consider the following valid inequality.
\begin{proposition}
Let $G$ be a graph, and let $K_{h}$ a clique of order $h \leq n$ of $G$. 
Then,
\begin{equation}\label{clique}
\begin{aligned}
\sum\limits_{v \in V(K_h)}{x_{v}} + \sum\limits_{e \in E(K_h)}{y_{e}} \leq \ceil{\frac{h}{2}}
\end{aligned}
\end{equation}
is a valid inequality for $P_{T}(G)$.
\end{proposition}
In particular, when the subgraph $K_h$ has even cardinality, we get the following result.
\begin{proposition}\label{EvenLemma}
Let $K_h$ %a graph, and let $K_h$ 
be a complete graph, where $h \in \mathbb{N}$ is an even number. Then, the {\bf even-clique inequality} defined as
\begin{equation}\label{even_clique}
\sum\limits_{v \in V(K_h)}{x_{v}} + \sum\limits_{e \in E(K_h)}{y_{e}} \leq \frac{h}{2}
\end{equation}
is facet-defining for the Total Matching Polytope $P_{T}(K_h)$.
\end{proposition}

\begin{proof}
Let $ G = K_h$ be a complete graph, where $h=2l$ for $l \in \mathbb{N}$ and let $V(K_h) :=\{v_1, v_2, \dots, v_{2l}\}$ and $E(K_h) :=\{e_{i,j}=\{v_i,v_j\} \mid \forall i,j \in \{1,2, \dots, 2l\}, i \neq j\}$.
First, we show that the even-clique inequalities are valid for $P_{T}(G)$.
Since $K_h$ is a complete graph of even order, it admits a perfect matching $M$.
Notice that any stable set $S$ intersects $K_h$ in at most one vertex, 
%
%so we don't  increase the cardinality of a total matching via adding a vertex %
thus a maximum total matching $T$ can be obtained by a perfect matching, 
or by deleting from a perfect matching an edge $e=\{i,j\}$ and adding one of its endpoints. 
This implies that $|T| \leq l$.
Next, we prove that the face induced by an even-clique inequality is facet-defining.
To this end, consider a face $F :=\{ z \in P_{T}(G) \mid \lambda^{T}z = \lambda_{0} \}$ and let $F':= \{z \in P_{T}(G) \mid \tilde\lambda^{T} z = \tilde\lambda_{0} \}$, where $\tilde\lambda^{T}z \leq \tilde\lambda_{0}$ corresponds to the even-clique inequality. 
Suppose that $F' \subseteq F$, we want to show that every inequality of $F$ is a scalar multiple of the even-clique inequality.
Place the vertices $v_1,v_2, \dots, v_{2l-1}$ at equal distances on a circle and place $v_{2l}$ in the center.
Starting from this configuration, we show a decomposition of $K_h$ into disjoint union of perfect matchings, such that $E(K_h)= M_1 \cup M_2 \cup \dots \cup M_{h-1}$.
Notice also that a perfect matching $M$ can be naturally identified as a total matching.
Now, fix an index $i$ and consider the edge that connects a vertex $v_{i}$ to the center $v_{2l}$ of the circle,
we call $c_{i} = \{v_i, v_{2l}\}$ the \textit{central edge}, and consider
the set of edges $ E_{i} :=\{e_{i+j,i-j} = \{v_{i+j},v_{i-j}\} \mid \forall j \in \{1, \dots, \frac{h}{2}-1\} \}$, where the indexes run modulo $h-1$.
It turns out that $M_{i} := E_{i} \cup \{c_{i}\}$ is a perfect matching.
In this way, repeating the same construction we can form $h-1$ distinct perfect matchings $M_{i}$, with $\chi[M_{i}] \in F'$, for all $i \in \{1,2, \dots, 2l-1\}$.
Now, we can construct a total matching with the same cardinality of the perfect matchings just constructed.
Consider an edge $e = \{v_{j},v_{k}\} \in M_{i}$ of a fixed perfect matching $M_{i}$.
Then, $T_{k} :=( M_{i} \setminus \{e_{j,k}\} ) \cup \{v_{k}\}$ and $T_{j} := (M_{i} \setminus \{e_{j,k}\} ) \cup \{v_j\}$ are total matchings.
Observe that $\chi[T_{j}] \in F'$ and $\chi[T_{k}] \in F'$, in particular these characteristic vectors lie on $F$.
This implies that 
%$\lambda_{v} = \lambda_{e} = 0$, for $v,e \in D \setminus \{M_{i},T_{j},T_{k}\}$ 
$\lambda_{v_{j}} = \lambda_{v_{k}} = \lambda_{e_{j,k}}$, since $\lambda^{T}\chi[T_j] = \lambda^{T}\chi[T_k] = \lambda^{T}\chi[M_i]$,
where we denote as $\lambda_{a}$ the cost coefficient for the element $a \in D = V \cup E$.
In particular, we apply this construction for all the edges of the same perfect matching $M_{i}$. 
Repeating the same argument for all the perfect matchings in the decomposition,
we obtain that $\lambda_v = \lambda_{e}$ for $e \in \delta(v)$,
$\forall v \in V$, and since the cost coefficients for the endpoints of each edge are the same by construction, and we consider only perfect matchings (we can touch each vertex),
we deduce that there exists $a \in \mathbb{R}$ such that $\lambda = a\mathbf{1}$.
Thus, this implies that $\lambda_0 = a \frac{h}{2}$.
We conclude that $\lambda^{T}z \leq \lambda_{0}$ is a scalar multiple of the even-clique inequality since $(\lambda, \lambda_0) = a(\mathbf{1},\frac{h}{2})$.
This completes the proof. \qed
\end{proof}

Now, we are ready to prove the main theorem of this section.
\begin{theorem}\label{kappaEven}
Let $G$ be a graph, and let $K_h$ 
be a complete graph, where $h$ is even. Then, the {\bf even-clique inequality} defined as
\begin{equation*}\label{even_clique_1}
\sum\limits_{v \in V(K_h)}{x_{v}} + \sum\limits_{e \in E(K_h)}{y_{e}} \leq \frac{h}{2}
\end{equation*}
is facet-defining for the Total Matching Polytope $P_{T}(G)$.
\end{theorem}

\begin{proof}
Let $K_h$ be a complete subgraph of even order of $G$.
We denote as $F$ the face induced by the even-clique inequality.
By Proposition \ref{EvenLemma}, we can find $|V(K_h)|+|E(K_h)|$ affinely independent points satisfying at equality the even-clique inequality.
Now, fix a perfect matching $M $ of $G[V(K_h)]$.
Since $M \cap \{u\} = \emptyset$ for every $u \notin V(K_h)$,
$T_u:= M \cup \{u\}$ is a total matching.
Observe that, $\chi[T_u] \in F$.
Thus, the set of characteristic vectors $\{\chi[T_u] \mid \forall u \notin V(K_h)\}$ is contained in $F$ and the corresponding $|V \setminus V(K_h)|$ points are affinely independent.
Clearly, it is easy to see that they are still affinely independent with respect to the previous points, so we have $n$ points up to now.
Similarly, $T_e:= M \cup \{e\}$ for every $e \notin \delta(V(K_h)) \cup E(K_h)$ is a total matching, since $M \cap \{e\} = \emptyset$.
Consequently, also the set of vectors $\{\chi[T_e] \mid \forall e \notin \delta(V(K_h) \cup E(K_h)\}$ is contained in $F$, and the corresponding points are affinely independent.
Now, let $S := \{ v \in V(K_h) \mid \delta(V(K_h)) \neq \emptyset\}$.
%
%As in the proof of proposition \ref{EvenLemma}, 
We can construct a total matching $T_{\overline{s}} :=( M_{s} \setminus \{e\}) \cup \{\overline{s}\} $, where $e=\{s,\overline{s}\}  \in E(K_h)$, $s \in S$ and $M_{s}$ is a perfect matching of $G[V(K_h)]$ with one end-point in $s$.
Then, $T_{\overline{e}_{s}} := T_{\overline{s}} \cup \{\overline{e}_{s}\}$ for every $\overline{e}_{s} \in \delta(V(K_h)) \cap \delta(s)$, is a total matching whose characteristic vector lies on $F$.
Repeating the same construction for all the edges $e_s \in \delta(V(K_h))$, we can obtain distinct total matchings for every $s \in S$ whose characteristic vectors belong to $F$, where the corresponding points are affinely independent.
In this way, we have found $n+m$ affinely independent points belonging to $F$, since we can rearrange the rows of the matrix having as columns these points in such a way that we get the following form:
\begin{center}
$\left[
\begin{array}{c|c|c}
A_{K_h} & B_{K_h} & C_{K_h} \\ \hline
\mathbf{0} & \widetilde{I_{v}} & \mathbf{0} \\ \hline
\mathbf{0} & \mathbf{0} & \widetilde{I_{e}}\\
\end{array}\right],
$
\end{center}
where the matrices $A_{K_h},B_{K_h},C_{K_h}$ have dimension $|V(K_h)| \times |E(K_h)|$ and correspond to the vertex and edge components of $K_h$.
The rest of the blocks are the zero and identity matrices of the remaining vertex and edge components.
This completes the proof. \qed
\end{proof} 
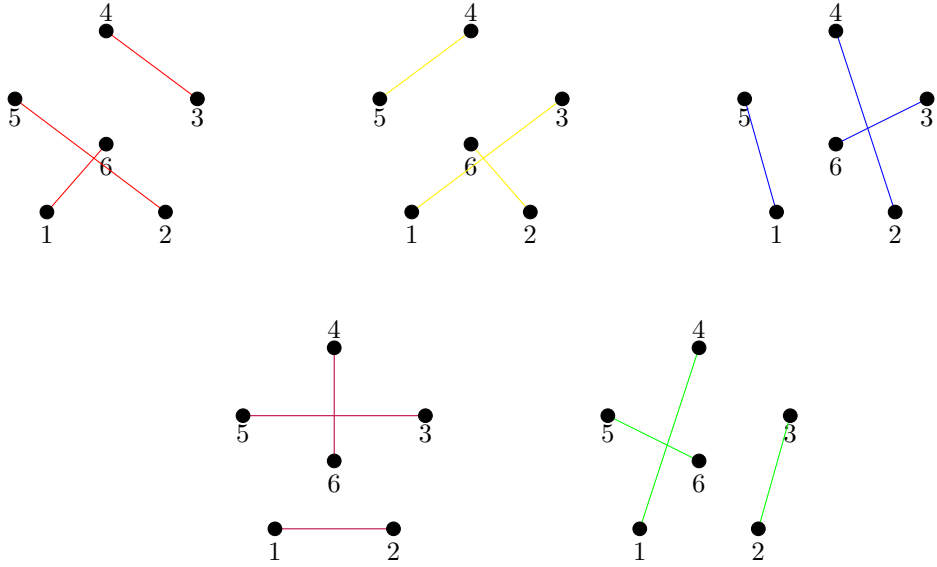
\begin{figure}[!htbp]
		\centering
		\begin{tikzpicture}[scale=0.60]
	    % first color
	
		\coordinate (v) at (1.2,1.5) ;
		\coordinate (w) at (3.8,1.5) ;
		\coordinate (z) at (4.5,4) ;
		\coordinate (x) at (0.5,4) ;
		\coordinate (c) at (2.5,3) ;
		\coordinate  (y) at (2.5,5.5);
		\node at (1.2,1) {$1$};
		\node at (3.8,1) {$2$};
		\node at (2.5,2.5){$6$};
		\node at (4.5,3.6) {$3$};
		\node at (2.5,5.9) {$4$};
		\node at (0.5,3.6) {$5$};
		\draw[red] (y)--(z);
	    \draw[red] (x)--(w);
	    \draw[red] (v)--(c);
		\draw[fill=black] (x) circle [radius=0.15cm];
		\draw[fill=black] (y) circle [radius=0.15cm];
		\draw[fill=black] (z) circle [radius=0.15cm];
		\draw[fill=black] (w) circle [radius=0.15cm];
		\draw[fill=black] (v) circle [radius=0.15cm];
		\draw[fill=black]  (c) circle [radius=0.15cm];

	    %second color
		\coordinate (v1) at (1.2+8,1.5) ;
		\coordinate (w1) at (3.8+8,1.5) ;
		\coordinate (z1) at (4.5+8,4) ;
		\coordinate (x1) at (0.5+8,4) ;
		\coordinate (c1) at (2.5+8,3) ;
		\coordinate  (y1) at (2.5+8,5.5);
		\node at (1.2+8,1) {$1$};
		\node at (3.8+8,1) {$2$};
		\node at (2.5+8,2.5){$6$};
		\node at (4.5+8,3.6) {$3$};
		\node at (2.5+8,5.9) {$4$};
		\node at (0.5+8,3.6) {$5$};
		\draw[yellow] (v1)--(z1);
		\draw[yellow] (y1)--(x1);
		\draw[yellow] (c1)--(w1);
		\draw[fill=black] (x1) circle [radius=0.15cm];
		\draw[fill=black] (y1) circle [radius=0.15cm];
		\draw[fill=black] (z1) circle [radius=0.15cm];
		\draw[fill=black] (w1) circle [radius=0.15cm];
		\draw[fill=black] (v1) circle [radius=0.15cm];
		\draw[fill=black]  (c1) circle [radius=0.15cm];

		%third color
		\coordinate (v2) at (1.2+16,1.5) ;
		\coordinate (w2) at (3.8+16,1.5) ;
		\coordinate (z2) at (4.5+16,4) ;
		\coordinate (x2) at (0.5+16,4) ;
		\coordinate (c2) at (2.5+16,3) ;
		\coordinate  (y2) at (2.5+16,5.5);
		\node at (1.2+16,1) {$1$};
		\node at (3.8+16,1) {$2$};
		\node at (2.5+16,2.5){$6$};
		\node at (4.5+16,3.6) {$3$};
		\node at (2.5+16,5.9) {$4$};
		\node at (0.5+16,3.6) {$5$};
		\draw[blue] (x2)--(v2);
		\draw[blue] (y2)--(w2);
		\draw[blue] (z2)--(c2);
		\draw[fill=black] (x2) circle [radius=0.15cm];
		\draw[fill=black] (y2) circle [radius=0.15cm];
		\draw[fill=black] (z2) circle [radius=0.15cm];
		\draw[fill=black] (w2) circle [radius=0.15cm];
		\draw[fill=black] (v2) circle [radius=0.15cm];
		\draw[fill=black]  (c2) circle [radius=0.15cm];

	    %four color
		\coordinate (v3) at (1.2+5,1.5-7) ;
		\coordinate (w3) at (3.8+5,1.5-7) ;
		\coordinate (z3) at (4.5+5,4-7) ;
		\coordinate (x3) at (0.5+5,4-7) ;
		\coordinate (c3) at (2.5+5,3-7) ;
		\coordinate  (y3) at (2.5+5,5.5-7);
		\node at (1.2+5,1-7) {$1$};
		\node at (3.8+5,1-7) {$2$};
		\node at (2.5+5,2.5-7){$6$};
		\node at (4.5+5,3.6-7) {$3$};
		\node at (2.5+5,5.9-7) {$4$};
		\node at (0.5+5,3.6-7) {$5$};
		\draw[purple] (w3)--(v3);
		\draw[purple] (x3)--(z3);
		\draw[purple] (y3)--(c3);
		\draw[fill=black] (x3) circle [radius=0.15cm];
		\draw[fill=black] (y3) circle [radius=0.15cm];
		\draw[fill=black] (z3) circle [radius=0.15cm];
		\draw[fill=black] (w3) circle [radius=0.15cm];
		\draw[fill=black] (v3) circle [radius=0.15cm];
		\draw[fill=black]  (c3) circle [radius=0.15cm];

		%five color
	    \coordinate (v4) at (1.2+13,1.5-7) ;
	    \coordinate (w4) at (3.8+13,1.5-7) ;
		\coordinate (z4) at (4.5+13,4-7) ;
		\coordinate (x4) at (0.5+13,4-7) ;
		\coordinate (c4) at (2.5+13,3-7) ;
		\coordinate  (y4) at (2.5+13,5.5-7);
		\node at (1.2+13,1-7) {$1$};
		\node at (3.8+13,1-7) {$2$};
		\node at (2.5+13,2.5-7){$6$};
		\node at (4.5+13,3.6-7) {$3$};
		\node at (2.5+13,5.9-7) {$4$};
		\node at (0.5+13,3.6-7) {$5$};
		\draw[green] (z4)--(w4);
		\draw[green] (y4)--(v4);
		\draw[green] (x4)--(c4);
		\draw[fill=black] (x4) circle [radius=0.15cm];
		\draw[fill=black] (y4) circle [radius=0.15cm];
		\draw[fill=black] (z4) circle [radius=0.15cm];
		\draw[fill=black] (w4) circle [radius=0.15cm];
		\draw[fill=black] (v4) circle [radius=0.15cm];
		\draw[fill=black]  (c4) circle [radius=0.15cm];
		\end{tikzpicture}
		
		\caption{Five perfect matchings of $K_{6}$}
		\label{fig:K6}
	\end{figure}

\begin{proposition}
Let $G$ be a graph, and let $K_h$ be a complete subgraph of $G$, where $h$ is odd. 
Then, the {\bf odd clique inequality}\label{odd_cli} defined as
\begin{equation}
\begin{aligned}
 \sum\limits_{v \in V(K_h)}{x_{v}} + \sum\limits_{e \in E(K_h)}{y_{e}} \leq \frac{h+1}{2}
\end{aligned}
\end{equation}
is valid for the Total Matching Polytope $P_{T}(G)$, but it is not facet-defining.
\end{proposition}

\begin{proof}
Let $K_{h}$ be a clique of odd order. Since in a total matching of $K_{h}$ we can pick at most one vertex, and the size of the largest matching is $\frac{h-1}{2}$,
we can take at most $\frac{h-1}{2}+1=\frac{h+1}{2}$ elements of a total matching, as shown in Figure \ref{fig:K5}.
Therefore, this implies that the odd clique inequality is valid for $P_{T}(G)$.
Now, we prove that it is not facet-defining.
Adding a vertex $u$ to the clique $K_{h}$, we can form a clique of even order $K_{h+1}:=(V(K_{h+1}),E(K_{h+1}))$, where $V(K_{h+1}):=V(K_{h}) \cup \{u\}$ and $E(K_{h+1}):= E(K_{h}) \cup \{e = \{u,v\} \mid v \in V(K_{h})\}$.
Then, the inequality 
\begin{equation*}
\sum\limits_{v \in V(K_{h})}x_{v} + \sum\limits_{e \in E(K_{h})}y_{e} \leq \frac{h+1}{2}
\end{equation*}
is dominated by
\begin{equation*}
\sum\limits_{v \in V(K_{h+1})}x_{v} + \sum\limits_{e \in E(K_{h+1})}y_e \leq \frac{h+1}{2}
\end{equation*}
This completes the proof.
\qed
\end{proof}

We stress that, even if the odd clique inequality is maximal it remains not facet-defining. 
Indeed, suppose that $K_h$ is a maximal clique of odd order.
We know that
\begin{align*}
    \sum_{v \in V(K_h)}x_v \leq 1
\end{align*}
 is facet-defining, and it is easy to notice that the following is a valid inequality for the Total Matching Polytope
\begin{align*}
    \sum_{e \in E(K_h)}y_e \leq \frac{h-1}{2}.
\end{align*}
Thus, the sum of the two inequalities gives the odd-clique inequality.
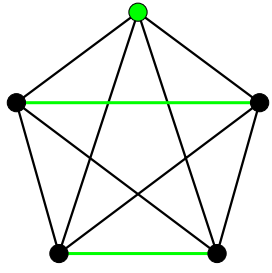
\begin{figure}[!t]
		\centering
		\begin{tikzpicture}[scale=0.40]
	    % first color
	    %second color
		
		\coordinate (v') at (2*1.2-2.5+16,2*1.5-3.2) ;
		\coordinate (w') at (2*3.8-2.5+16,2*1.5-3.2) ;
		\coordinate (z') at (2*4.5-2.5+16,2*4-3.2) ;
		\coordinate (x') at (2*0.5-2.5+16,2*4-3.2) ;
		\coordinate  (y') at (2*2.5-2.5+16,2*5.5-3.2);
	    
	    \draw[line width=0.03 cm] (z')--(y')--(v')--(x')--(y')--(w')--(z')--(v');
	    \draw[line width=0.03 cm] (w')--(x');
	    \draw[green,line width=0.04 cm] (v')--(w');
	    \draw[green,line width=0.04 cm] (x')--(z');
	    
	    \draw[fill=black]  (v') circle [radius=0.30cm];
		\draw[fill=black] (w') circle [radius=0.30cm];
		\draw[fill=black] (z') circle [radius=0.30cm];
		\draw[fill=green] (y') circle [radius=0.30cm];
		\draw[fill=black] (x') circle [radius=0.30cm];

		\end{tikzpicture}
		\caption{A complete $K_5$ graph. In green, a possible maximal total matching.\label{fig:K5}}
\end{figure}

%%%%%%%%%%%%%%%%%%
\paragraph{Separation for the even-clique inequalities}

We propose the following ILP model to detect a maximum violated even-clique, which is based on the maximum edge weighted clique model discussed in \cite{MaxWeightedClique}:
\begin{align}
\label{cliq:obj}  \max \quad & \sum\limits_{v \in V}c_vx_{v} + \sum\limits_{e \in E}w_e y_{e} - z\\
\label{cliq:c1}    \mbox{s.t.} \quad 
                 &  x_v + x_w \leq 1 & \forall \{v,w\} \in \overline{E} \\
\label{cliq:c2}  & \sum_{v \in V}x_v  =2z  \\
\label{cliq:c3}  &  y_e \leq x_v & \forall e =\{u,v\} \in E \\
\label{cliq:c4}  &  y_e \leq x_u & \forall e= \{u,v\}  \in E \\
\label{cliq:c5}  &  x_v + x_u \leq y_e + 1 & \forall e = \{u,v\} \in E \\
\label{cliq:c6}  & x_v,y_e \in \{0,1\} & \forall v \in V, \forall e \in E \\
\label{cliq:c7}  & z \in \mathbb{Z},
\end{align}
\noindent where $\overline{E}$ represents the complement of $E(G)$. Since we want to detect a clique of even order, we introduce
the integer variable $z \in \mathbb{Z}$.
If the optimal solution is positive, we get a maximally violated even-clique inequality.
The constraints \eqref{cliq:c1} are equivalent to imposing the condition that we can select at most one vertex from a maximal stable set in the clique found.
%
%Notice that there are exponentially (in the size of the graph) many constraints of this type, and their separation is NP-hard.
%
In \cite{MaxWeightedClique}, it is proven that finding a maximum weighted edge clique is NP-hard.
Consequently, the problem \eqref{cliq:obj}--\eqref{cliq:c6} is NP-hard in general, and it contains the maximum edge weighted clique as a special case. %\cite{MaxWeightedClique}.}
%
%The NP-hardness of problem \eqref{cliq:obj}--\eqref{cliq:c6} allows us to give a polyhedral proof of the NP-hardness also of the Weighted Total Matching Problem.

% NON SERVE PIU' E' GIA' STATO DIMOSTRATO E QUESTA E' BANALE
\begin{comment}
Now, we prove that the Weighted Total Matching Problem is an NP-hard problem

%
% Questo va detto nell'introduzione:
% different from the proof presented for instance in \citep{NphardTotalmatching}.
\begin{theorem}
The Weighted Maximum Total Matching Problem is NP-hard.
\end{theorem}
\begin{proof}
\revision{
Consider the optimization problem $\nu(G,\alpha,\beta)$ for $P_{T}(G)$.
%
Let $\beta_e = 0$ for every $e \in E$ be an instance of $\nu(G,\alpha,\beta)$.
%
Since The Weighted Maximum Total Matching contains the Stable Set Problem, which is an NP-hard problem, the claim follows.
%Since the Total Matching Polytope is full-dimensional (Proposition \ref{fulldim}) and the separation problem for the even-clique inequality is NP-hard (Theorem \ref{kappaEven}), by applying the Equivalence Theorem between Optimization and Separation, (see \cite{SeparationOptimization}, Chap. 6.4, pp. 174--181), we conclude that solving problem $\nu(G,\alpha,\beta)$ is NP-hard.
\qed
}
\end{proof}

\end{comment}

%%%%%%%%%%%%%%%%%
%\begin{comment}

\section{Computational results}

%{\bf TODO: i risultati computazionali vanno rafforzati. Usare random casuali a densità crescente (sui grafi sparsi i vincoli di ciclo dovrebbero dare un contributo importante). Usare bene le istanze DIMACS, con test più consistenti.

%}
This section presents computational results for the two total coloring relaxations presented in Section \ref{sec:coloring} and for the total matching relaxations based on different valid (facet) inequalities discussed in Sections \ref{sec:matching} and \ref{sec:valid}.
The results for the TCP are obtained with a Column Generation algorithm based on model \eqref{m3:d1}--\eqref{m3:d3}, which are compared with the results provided by the LP assignment model \eqref{m1:obj}--\eqref{m1:y}.
The results for the TMP aim to compare the strength of the families of valid (facet) inequalities discussed in this paper.
For both problems, the goal of the computational tests is to compare the bound strengths, that are, lower bounds for total coloring and upper bounds for total matching.
%
%More sophisticated implementations could improve the speed of the algorithms, but will not tighten the bounds obtained at the root node.
\paragraph{Datasets} 
First, we tested our algorithms on a few named graphs: the cycle $C_5$, the complete graph $K_{12}$, and the classical Petersen, Chvatal, Tutte, and Watkins graphs.
Then, after running preliminary tests on a large set of graphs from the literature, we decided to focus on a small set of graphs to evaluate the total coloring relaxations. 
In practice, we observed that most of the graphs from the literature are of Type-1, that is, they have $\chi_T(G)=\Delta(G)+1$ (see, \cite{Yap2006}).
For this class of graphs, the naive lower bound equal to $\Delta(G)+1$ is tight, and the contribution of any LP relaxation is minimal. 
Hence, we have selected 11 cubic graphs of Type-1 and 11 Snark graphs of Type-2, those having $\chi_T(G)>\Delta(G)+1$, downloaded from the House of Graph library\footnote{See \url{https://hog.grinvin.org/}}, first introduced in \cite{Brinkmann2015}, and named {\tt graph\_N}. 
The Snarks graphs are cyclically 4-edge-connected graphs with $\chi'(G)= 4$.
For Type-2 graphs, we show that the set covering model \eqref{m3:d1}--\eqref{m3:d3} provides better lower bounds than the assignment model \eqref{m1:obj}--\eqref{m1:y}.
Finally, we consider random cubic graphs of different sizes and random sparse graphs with 80 vertices but different edge densities for evaluating the total matching relaxations.
%
%We show which family of inequalities is more effective for a given graph structure by changing the graph density.

\paragraph{Implementation Details} We have implemented in Python a Column Generation algorithm for the TCP and a Cutting Plane algorithm for the TMP. We use Gurobi 9.1.1 for solving both the master, the pricing, and the different cut-separation subproblems.
The experiments are run on a Dell Workstation with a Intel Xeon W-2155 CPU with 10 physical cores at 3.3GHz and 32 GB of RAM.
The source code and the dataset is freely available on GitHub at \url{https://github.com/stegua/total-matching}.

%%%%%%%%%%%%%%%%%%%%%%%%%%%%%%%%%%%%%%%%%%%%%%%%%%%%%%%%%%%%%%%%%%%%%%%%
\subsection{Total Coloring Lower Bounds}

Table \ref{tab:my_label} reports our results for comparing the bound strength achieved by the LP relaxation of model \eqref{m2:obj}--\eqref{m2:l} (LP, 6-$th$ column) and the LP relaxation of the set covering model (SC-LP, 7-$th$ columnn). 
The first columns give for each graph the number of vertices $n$ and of edges $m$, the maximum degree $\Delta(G)$, and the total chromatic number $\chi_T(G)$.
The last two columns report the number of column generation iterations (CG iter) and the total runtime in seconds.
For all Type-1 graphs, the trivial lower bound $\Delta(G)+1$ is already equal to $\chi_T(G)$, and the contribution of the LP relaxation in terms of bounds is null. 
On the contrary, for all Type-2 graphs, the lower bounds provided by the column generation algorithm are higher than those obtained with the assignment model: this provides computational evidence that the inequality $z_{LP}^{(C)} \geq z_{LP}^{(A)}$ in Proposition 3 is not always tight.

\begin{table}[!t]
    \centering
    \begin{tabular}{l@{\hskip 0.15in}r@{\hskip 0.15in}r@{\hskip 0.15in}r@{\hskip 0.15in}r@{\hskip 0.15in}r@{\hskip 0.15in}r@{\hskip 0.15in}r@{\hskip 0.15in}r@{\hskip 0.15in}r}
    	Graph Name & $n$	& $m$&	Type & $\Delta(G)$ &	$\chi_T(G)$ &	LP  &	SC-LP & CG iter & Runtime \\
    	\hline
Cycle $C_5$ &	5&	5&	Type-2 & 2&	4&	3.00&	{\bf 3.33}&	22&	0.00\\
Complete $K_{12}$&	12&	66&	Type-2 &11&	13&	12.00&	{\bf 13.00}&	156&	0.38\\
Petersen&	10&	15&Type-1 &	3&	4&	4.00&	4.00&	61&	0.00\\
Chvatal&	12&	24&Type-1 &	4&	5&	5.00&	5.00&	95&	0.01 \\
Tutte &	46&	69&Type-1 &	3&	4&	4.00&	4.00& 777	&	11.35 \\
Watkins	&	50	&	75	&Type-1 &	3	&	4	&	4.00	&	4.00	&	686	&	9.10	\\
%3-Regular	&	24	&	36	&	3	&	4	&	4.00	&	4.00	&	201	&	0.23	\\
%3-Regular	&	50	&	75	&	3	&	4	&	4.00	&	4.00	&	750	&	11.11	\\
%Random	&	50	&	939	&	42	&	43	&	43.00	&	43.00	&	1239	&	3.82	\\
\hline
graph\_6921  & 20 & 30 &Type-1 & 3 & 4 & 4.00 & 4.00 & 164  & 0.08 \\
graph\_1008  & 22 & 33 &Type-1 & 3 & 4 & 4.00 & 4.00 & 219  & 0.17 \\
graph\_1012  & 22 & 33 &Type-1 & 3 & 4 & 4.00 & 4.00 & 207 & 0.16 \\
graph\_3334  & 26 & 39 &Type-1 & 3 & 4 & 4.00 & 4.00 & 279  & 0.37 \\
graph\_20015 & 30 & 45 &Type-1 & 3 & 4 & 4.00 & 4.00 & 314  & 0.56 \\
graph\_3383  & 36 & 54 &Type-1 & 3 & 4 & 4.00 & 4.00 & 438  & 1.50 \\
graph\_22470 & 38 & 57 &Type-1 & 3 & 4 & 4.00 & 4.00 & 459  & 1.83 \\
graph\_25159 & 44 & 66 &Type-1 & 3 & 4 & 4.00 & 4.00 & 583  & 6.44 \\
graph\_1338  & 50 & 75 &Type-1 & 3 & 4 & 4.00 & 4.00 & 691 & 7.01 \\
graph\_1427  & 50 & 75 &Type-1 & 3 & 4 & 4.00 & 4.00 & 704  & 7.30 \\
graph\_1389  & 60 & 90 &Type-1 & 3 & 4 & 4.00 & 4.00 & 1010 & 21.15 \\
\hline
graph\_6630	&	22	&	31	&Type-2 &	3	&	5	&	4.00	&	{\bf 4.08}	&	185	&	0.14	\\
graph\_6710	&	40	&	60	&Type-2 &	3	&	5	&	4.00	&	{\bf 4.08}	&	459	&	2.70	\\
graph\_6714	&	40	&	60	&Type-2 &	3	&	5	&	4.00	&	{\bf 4.08}	&	429	&	2.19	\\
graph\_6720	&	40	&	60	&Type-2 &	3	&	5	&	4.00	&	{\bf 4.08}	&	444	&	2.89	\\
graph\_6724	&	40	&	60	&Type-2 &	3	&	5	&	4.00	&	{\bf 4.08}	&	441	&	2.60	\\
graph\_6728	&	40	&	60	&Type-2 &	3	&	5	&	4.00	&	{\bf 4.08}	&	458	&	2.85	\\
graph\_6708	&	40	&	60	&Type-2 &	3	&	5	&	4.00	&	{\bf 4.08}	&	456	&	3.03	\\
graph\_6712	&	40	&	60	&Type-2 &	3	&	5	&	4.00	&	{\bf 4.08}	&	439	&	2.67	\\
graph\_6718	&	40	&	60	&Type-2 &	3	&	5	&	4.00	&	{\bf 4.08}	&	445	&	2.66	\\
graph\_6722	&	40	&	60	&Type-2 &	3	&	5	&	4.00	&	{\bf 4.08}	&	435	&	2.51	\\
graph\_6726	&	40	&	60	&Type-2 &	3	&	5	&	4.00	&	{\bf 4.08}	&	456	&	2.80	\\
\hline\\
    \end{tabular}
    \caption{Comparing total coloring lower bounds obtained with two different relaxations: LP refers to the relaxation of \eqref{m2:obj}--\eqref{m2:l}, while SC-LP refers to \eqref{m3:d1}--\eqref{m3:d3}.}
    \label{tab:my_label}
\end{table}

\subsection{Total Matchings Lower Bounds on Snark Graphs}

Table~\ref{tab:snark} presents the computational results on the total matching problem for 11 snark (cubic) graphs of Type-2.
For each graph, the table reports the number of vertices $n$ and of edges $m$, the matching number $\nu(G)$, the stable set number $\alpha(G)$, and the total matching number $\nu_T(G)$.
Then, we report the upper bounds obtained with the {\it basic inequalities} \eqref{b1:v}--\eqref{b1:a} (column {\it Basic}), the upper bounds obtained separating only vertex-clique inequalities \eqref{vertex-clique} (column {\it Clique}), and separating only congruent-$2k3$ cycle inequality \eqref{cycle} (column {\it Cycle-$2k3$}).
For the later inequality, we also report the percentage of the optimality gap, computed as $\frac{UB-\nu_T(G)}{\nu_T(G)}\times 100$.
Since Snark graphs have no cliques violated, the only inequalities that improve the bounds are the congruent-$2k3$ cycle inequalities.
Notice that we did not even try to separate even-clique inequalities for this family of graphs since Snark graphs are a special case of cubic graphs, and they cannot have cliques with cardinality larger than three.
The results presented in the next paragraphs will show for which type of graphs the vertex-clique and the even-clique inequalities begin to play a role.

\begin{table}[t!]
    \centering
    \begin{tabular}{@{ }lc@{ }c@{ }c@{ }c@{ }cc@{ }c@{ }c@{ }cc@{ }c@{ }}
    \hline
            \multicolumn{6}{c}{Snark Graphs} & \multicolumn{4}{c}{Upper bounds} & \multicolumn{2}{c}{Number of cuts} \\  
Name & $n$	&	$m$	&	$\nu(G)$	&	$\alpha(G)$	&	$\nu_T(G)$	&	Basic	&	Clique	&	Cycle-$2k3$	&	Percentage & Clique	&	Cycle-$2k3$ \\
\hline
graph\_6630 & 22	&	31	&	11	&	9	&	13	&	15.23	&	15.23	&	14.24	&	9.5\%	&	0	&	19	\\
graph\_6710 & 40	&	60	&	20	&	17	&	26	&	28.00	&	28.00	&	27.24	&	4.8\%	&	0	&	24	\\
graph\_6714 & 40	&	60	&	20	&	16	&	26	&	28.00	&	28.00	&	27.13	&	4.3\%	&	0	&	26	\\
graph\_6720 & 40	&	60	&	20	&	16	&	26	&	28.00	&	28.00	&	27.23	&	4.7\%	&	0	&	25	\\
graph\_6724 & 40	&	60	&	20	&	16	&	26	&	28.00	&	28.00	&	27.21	&	4.6\%	&	0	&	27	\\
graph\_6728 & 40	&	60	&	20	&	16	&	26	&	28.00	&	28.00	&	27.20	&	4.6\%	&	0	&	22	\\
graph\_6708 & 40	&	60	&	20	&	17	&	26	&	28.00	&	28.00	&	27.20	&	4.6\%	&	0	&	29	\\
graph\_6712 & 40	&	60	&	20	&	16	&	26	&	28.00	&	28.00	&	27.20	&	4.6\%	&	0	&	29	\\
graph\_6718 & 40	&	60	&	20	&	16	&	26	&	28.00	&	28.00	&	27.19	&	4.6\%	&	0	&	24	\\
graph\_6722 & 40	&	60	&	20	&	16	&	26	&	28.00	&	28.00	&	27.14	&	4.4\%	&	0	&	27	\\
graph\_6726 & 40	&	60	&	20	&	16	&	26	&	28.00	&	28.00	&	27.21	&	4.7\%	&	0	&	32	\\
\hline
    \end{tabular}
    \caption{Total Matching results for 11 Snark (cubic) graphs: Upper bounds and number of generated cuts.}
    \label{tab:snark}
\end{table}

%%%%%%%%%%%%%%%%%%%%%%%%%%%%%%%%%%%%%%%%%%
\subsection{Total Matchings Lower Bounds on Cubic Graphs}

Table~\ref{tab:cubic:single} presents the computational results on 6 random cubic graphs.
For each graph, the table gives the number of nodes $n$ and of edges $m$, the matching number $\nu(G)$, the stable set number $\alpha(G)$, and the total matching number $\nu_T(G)$.
From the 6-th to the 9-th columns we report the upper bound obtained with the LP relaxation \eqref{b1:v}--\eqref{b1:a} (column {\it Basic}), starting with \eqref{b1:v}--\eqref{b1:a} and separating only constraints \eqref{vertex-clique} (column {\it Clique}), starting with \eqref{b1:v}--\eqref{b1:a} and separating only constraints \eqref{cycle} (column {\it Cycle-$2k3$}), 
and starting with \eqref{b1:v}--\eqref{b1:a} and separating both \eqref{vertex-clique} and \eqref{cycle} (column {\it All}). 
For the same combinations of inequalities, the last four columns report the percentage optimality gap computed with respect to $\nu_T(G)$.

We remark that on random cubic graphs, which are very sparse graphs, the congruent-$2k3$ cycle inequalities close a larger fraction of the optimality gap than the vertex-clique inequalities.
% 
%This is expected since on cubic graphs we can have at most cliques of size three.
%
%However, by separating both vertex-clique and congruent-$2k3$ cycle inequalities, we can obtain even smaller optimality gaps.

Table~\ref{tab:cubic} reports the average results for our  algorithm based on vertex-clique constraints and congruent-$2k3$ cycle inequalities.
%
%Notice that on cubic graphs with cannot have any even-clique at all.
%
The purpose of the table is to show the average strength in terms of bound strength for both families of inequalities.
For each row of the table, we report the average over 10 random instances of the number of violated cuts identified by our algorithm and of the percentage gap with respect to the optimal solution.
We remark that the number of vertex-cliques is limited, with an average of violated cuts ranging 0.9 to to 2.1.
The number of violated congruent-$2k3$ cycle inequalities is in average much larger, but the percentage gap of the upper bound is only slightly better.
Again, combining both vertex-cliques and congruent-$2k3$ cycle inequalities, the average percentage gap of the upper bound is always stronger than for the two families taken separately.

\begin{table}[t!]
    \centering
    \begin{tabular}{@{ }r@{ }rc@{ }c@{ }ccccccccc@{ }}
    \hline
        \multicolumn{5}{c}{Cubic Graphs} & \multicolumn{4}{c}{Upper bounds} & \multicolumn{4}{c}{Percentage optimality gap} \\  
$n$&	$m$	&$\nu(G)$	&$\alpha(G)$	& $\nu_T(G)$& Basic	& Clique	& Cycle-$2k3$ & All & Basic & Clique & Cycle-$2k3$ & All \\
\hline
50	& 75	& 25	& 22	& 34	& 35.00	& 34.92	& 34.41	& 34.33	& 2.94\%	& 2.70\%	& 1.20\% & 0.98\% \\
60	& 90	& 30	& 26	& 40	& 42.00	& 42.00	& 41.27	& 41.27	& 5.00\%	& 5.00\%	& 3.18\% & 3.18\% \\
70	& 105	& 35	& 30	& 47	& 49.00	& 48.95	& 48.28	& 48.25	& 4.26\%	& 4.15\%	& 2.71\% & 2.66\% \\
80	& 120	& 40	& 35	& 54	& 56.00	& 55.95	& 55.51	& 55.46	& 3.70\%	& 3.61\%	& 2.79\% & 2.70\% \\
90	& 135	& 45	& 39	& 61	& 63.00	& 62.78	& 62.29	& 62.05	& 3.28\%	& 2.91\%	& 2.12\% & 1.73\% \\
100	& 150	& 50	& 44	& 68	& 70.00	& 69.85	& 69.08	& 69.03	& 2.94\%	& 2.72\%	& 1.58\% & 1.52\% \\
\hline
    \end{tabular}
    \caption{Total Matching results for six cubic graphs: Comparison of upper bounds  and of the percentage optimality gaps.}
    \label{tab:cubic:single}
\end{table}

\begin{table}[t!]
    \centering
    \begin{tabular}{@{ }r@{ }rcccccccc@{ }}
    \hline
    \multicolumn{3}{c}{Cubic Graphs} & \multicolumn{3}{c}{Mean of violated cuts} & \multicolumn{4}{c}{Mean of percentage gap} \\  
$n$ &	$m$ &	density&	Clique &	Cycle-$2k3$	& All	& Basic & Clique	& Cycle-$2k3$	& All \\
\hline
50	&75	&6.1\%	&0.9	&24.9	&25.1	&3.88\%	&3.75\%	&2.14\% & 2.06\% \\
60	&90	&5.1\%	&1.3	&25.3	&27.4	&3.72\%	&3.54\%	&2.20\% & 2.10\% \\
70	&105	&4.3\%	&2.1	&27.3	&27.6	&4.04\%	&3.81\%	&2.79\% & 2.65\% \\
80	&120	&3.8\%	&1.8	&28.1	&29.9	&3.70\%	&3.52\%	&2.55\% & 2.44\% \\
90	&135	&3.4\%	&1.5	&30.1	&32.3	&3.45\%	&3.32\%	&2.21\% & 2.11\% \\
100	&150	&3.0\%	&1.8	&30.3	&33.0	&3.09\%	&2.95\%	&2.05\% & 1.96\% \\
\hline
    \end{tabular}
    \caption{Comparison of maximal clique inequality and congruent-$2k3$ cycle inequalities: Average results of the number of violated cuts and percentage optimality gap. Each row reports the average over 10 random instances of the same size.}
    \label{tab:cubic}
\end{table}

%%%%%%%%%%%%%%%%%%%%%%%%%%%%%%%%%%%%%%%%%%
\subsection{Total Matchings Lower Bounds on Random Graphs}
We finally present the results of total matching problems comparing the LP relaxations based on different valid inequalities. 
We use random graphs with 80 vertices and an edge density ranging from 5\% up to 25\%. 
We do not report results for larger edge density because the main pattern on the cut strength does not change anymore.
Table~\ref{tab:random:single} reports in the first column the percentage graph density.
In the remaining columns, the table gives the numbers for $\nu(G)$, $\alpha(G)$, and $\nu_T(G)$. 
Regarding the upper bounds and the percentage optimality gap, in addition to the vertex-clique and the congruent-$2k3$ inequalities, we also consider the even-clique inequalities \eqref{even_clique_1} (column {\it Even-C}).
We notice that on very sparse graphs, the congruent-$2k3$ cycle inequalities close a large fraction of the optimality gap.
However, as soon as the graphs become denser, the vertex-clique inequalities play a crucial role in reducing the upper bounds (and hence reducing the percentage optimality gap).
On dense graphs, the even-clique inequalities reduce the upper bounds, but only marginally, and they are not as effective as the vertex-clique inequalities, despite being facet-defining.
%
%Indeed, when we separate all the families of valid inequalities together, we get the best upper bounds on all five random graphs.

Table~\ref{tab:random} reports more extensive results for random graphs: in each row, the table gives the average over 10 random graphs with the same density the number of violated cuts and the percentage optimality gap.
%
%As already observed for cubic graphs, the number of violated congruent-$2k3$ cycle inequalities is large, despite the small fraction of the closed gap.
%
For very sparse graphs, combining several families of inequalities pays off in terms of upper bound.
However, on dense random graphs, it appears that only the vertex-clique inequalities are very effective in reducing the optimality gap.
For instance, on graphs with an edge density of 25\% (last row of Table \ref{tab:random}), by separating a mean of 131.1 cuts, we get a mean optimality gap of 6.3\%.
At the same time, we have an average of 831.7 violated congruent-$2k3$ cycle inequalities achieving an optimality gap of 25.22\%, and an average of 232.5 even-clique violated inequalities for a gap of 25.64\%.
Indeed, vertex-clique inequalities are the most effective in reducing the upper bounds.

\begin{table}[t!]
    \centering
    \begin{tabular}{@{ }c@{ }c@{ }c@{ }c@{ }cc@{ }c@{ }c@{ }cc@{ }c@{ }c@{ }c@{ }c@{ }}
    \hline
    \multicolumn{4}{c}{Random Graphs} & \multicolumn{5}{c}{Upper bounds} & \multicolumn{5}{c}{Percentage optimality gap} \\  
den	&	$\nu(G)$	&	$\alpha(G)$	&	$\nu_T(G)$	&	Basic	&	Clique	&	Cycle-$2k3$	&	Even-C	&	All	&	Basic & Clique	&	Cycle-$2k3$	&	Even-C	&	All	\\
\hline
5\%	&	39	&	40	&	58	&	58.90	&	58.89	&	58.79	&	58.90	&	58.78	&	1.54\%	&	1.54\%	&	1.37\%	&	1.54\%	&	1.34\%	\\
10\%	&	40	&	28	&	54	&	58.12	&	55.56	&	57.14	&	58.12	&	55.52	&	7.63\%	&	2.89\%	&	5.82\%	&	7.63\%	&	2.81\%	\\
15\%	&	40	&	21	&	50	&	58.95	&	53.23	&	57.76	&	58.90	&	53.23	&	17.90\%	&	6.47\%	&	15.51\%	&	17.80\%	&	6.47\%	\\
20\%	&	40	&	18	&	49	&	59.28	&	51.16	&	58.49	&	59.04	&	51.16	&	20.98\%	&	4.42\%	&	19.36\%	&	20.49\%	&	4.42\%	\\
25\%	&	40	&	16	&	48	&	59.41	&	49.97	&	58.78	&	59.02	&	49.97	&	23.78\%	&	4.11\%	&	22.46\%	&	22.96\%	&	4.11\%	\\
\hline
    \end{tabular}
    \caption{Total Matching results for five random graphs with 80 vertices and different edge density (column {\it den}): Comparison of upper bounds  and of percentage optimality gaps.}
    \label{tab:random:single}
\end{table}

\begin{table}[t!]
    \centering
    \begin{tabular}{rrrrrrrrrr}
\hline
\multicolumn{1}{c}{$\,$} & \multicolumn{4}{c}{Mean of violated cuts} & \multicolumn{5}{c}{Mean of percentage gap} \\
den	&	Clique	&	Cycle-$2k3$	&	Even-C	&	All	&	Basic	&	Clique	&	Cycle-$2k3$	&	Even-C	&	All	\\
\hline
5\%	&	2.3	&	14.8	&	0.0	&	12.1	&	1.08\%	&	0.86\%	&	0.86\%	&	1.08\%	&	0.86\%	\\
10\%	&	44.1	&	166.6	&	0.8	&	112.5	&	9.52\%	&	4.61\%	&	8.06\%	&	9.52\%	&	4.55\%	\\
15\%	&	95.5	&	590.8	&	15.9	&	137.9	&	17.93\%	&	5.74\%	&	15.56\%	&	17.68\%	&	5.74\%	\\
20\%	&	117.0	&	669.1	&	75.3	&	142.6	&	20.96\%	&	4.61\%	&	19.31\%	&	20.48\%	&	4.61\%	\\
25\%	&	131.1	&	831.7	&	232.5	&	153.0	&	26.48\%	&	6.30\%	&	25.22\%	&	25.64\%	&	6.30\%	\\
\hline
    \end{tabular}
    \caption{Comparison of the strength of valid inequalities for total matching: Average results of the number of violated cuts and percentage optimality gaps for random graphs with 80 vertices and different edge density. Each row reports the average over 10 random instances of the same size.}
    \label{tab:random}
\end{table}

%%%%%%%%%%%%%%%%%%%%%
\section{Conclusion and future works}
In this paper, we have proposed polyhedral approaches to the Total Coloring and Total Matching problems.
We have introduced two ILP models for the TCP: the assignment model and the set covering model based on maximal total matchings, and we have shown how to get the second model by applying a Dantzig-Wolfe reformulation to the first.
For the TMP, our main contributions include a partial characterization of the feasible region of the Total Matching Polytope and the introduction of two families of nontrivial valid inequalities: congruent-$2k3$ cycle inequalities that are facet-defining when $k=4$, and the facet-defining even-clique inequalities.
Curiously, we have also shown that the odd clique inequalities are valid, but they are not facet-defining.
%Finally, our polyhedral study of the Total Matching Problem has permitted an alternative proof for the NP-hardness of the Weighted Total Matching Problem.

As future work, we plan to give a complete description of the Total Matching Polytope for certain classes of graphs.
For instance, since we have a complete description of the Stable Set Polytope for bipartite graphs, our research direction will be to study new facet-defining inequalities that will completely describe the Total Matching Polytope for bipartite graphs,
and likely, for other known classes as quasi-line and claw-free graphs, \cite{quasi-line,Yuri}.
Computationally, it will be of interest, as future development, to implement a complete branch-and-price algorithm for the Total Coloring Problem, and a complete branch-and-cut algorithm for the TMP.

%%%%%%%%%%%%%%%%%%%%%%%%%%%%%%%%%%%%%%%%%%%%%%%
\section*{Acknowledgments}
The research was partially supported by the Italian Ministry of Education, University and Research (MIUR): Dipartimenti di Eccellenza Program (2018--2022) - Dept. of Mathematics ``F. Casorati'', University of Pavia.

%   We are deeply indebted to Stefano Coniglio for discussions on the equivalence between optimization and separation as presented in \cite{SeparationOptimization}.

%
%Moreover, since exact methods of branch-and-cut have been proposed for the Stable Set Problem (see \cite{RebennackStable}), our intention is to set up a branch-and-cut algorithm to tackle the Maximum Total Matching Problem.
%, and, hopefully, for complete graphs. 

%
% STEGUA: UN PO' RIDONANTE!
%It turns out (DOVE?) that adding valid inequalities corresponding to complete bipartite graphs, the non integral solutions are cut off.
%
%Moreover, since our perspective is to build a set covering formulation for the Total Coloring problem by maximal total matchings, we want to study also the feasible region associated to the Maximal Total Matching Polytope. 
%
% (SI RIPETE CON QUANTO PRIMA) Finally, we plan to characterize the Maximal Total Matching Polytope for special classes of graphs, such as complete bipartite graphs and complete graphs.
%
%Finally, we want to set up an efficient branch-and-price algorithm to tackle the set covering formulation for the Total Coloring Problem by a column generation.

\bibliographystyle{apalike}
\bibliography{references}

\end{document}